\documentclass[12pt,oneside,a4paper]{article}
\usepackage{graphicx}
\usepackage{amssymb}
\usepackage{amsthm}
\usepackage{aliascnt}
\usepackage{amsmath}
\usepackage{mathtools}
\usepackage{tikz}
\usetikzlibrary{shapes}
\usetikzlibrary{patterns}
\usetikzlibrary{arrows,positioning,decorations.pathmorphing,trees}
\usetikzlibrary{arrows.meta}
\usepackage{cancel}
\usepackage{soul}
\usepackage{tikz-cd}
\makeatletter
\tikzset{
  column sep/.code=\def\pgfmatrixcolumnsep{\pgf@matrix@xscale*(#1)},
  row sep/.code   =\def\pgfmatrixrowsep{\pgf@matrix@yscale*(#1)},
  matrix xscale/.code=
    \pgfmathsetmacro\pgf@matrix@xscale{\pgf@matrix@xscale*(#1)},
  matrix yscale/.code=
    \pgfmathsetmacro\pgf@matrix@yscale{\pgf@matrix@yscale*(#1)},
  matrix scale/.style={/tikz/matrix xscale={#1},/tikz/matrix yscale={#1}}}
\def\pgf@matrix@xscale{1}
\def\pgf@matrix@yscale{1}
\makeatother
\usepackage{stackengine}
\usepackage{pgfplots}
\pgfplotsset{compat=1.18}
\usepgfplotslibrary{groupplots}
\usepackage{float}
\usepackage{braket}
\usepackage{setspace}
\usepackage{fullpage}
\usepackage{newpxtext,newpxmath}
\usepackage{comment}
\usepackage{enumitem}
\usepackage[round]{natbib}
\usepackage{hyperref}
\usepackage[capitalize]{cleveref}
\hypersetup{
    colorlinks=true,
    linkcolor=black,
    citecolor=black
    }

\usepackage{manyfoot}
\SetFootnoteHook{\hspace*{-1.8em}}
\DeclareNewFootnote{bl}[gobble]
\usepackage{algorithm}
\usepackage{algpseudocodex}

\newcommand{\mathmagnifier}{\mathord{%
  \tikz[baseline=-0.6ex, x=1ex,y=1ex]{
    \draw[line width=0.12ex] (0,0) circle (0.4);
    \draw[line width=0.12ex] (0.28,-0.28) -- (0.84,-0.84);
  }%
}}

\newtheorem{theorem}{Theorem}

\newtheorem{corollary}{Corollary}

\theoremstyle{definition}
\newtheorem{fact}{Fact}
\newtheorem{reduction}{Reduction}

\newtheorem{lemma}{Lemma}

\newaliascnt{proposition}{theorem} 
\newtheorem{proposition}[proposition]{Proposition}
\aliascntresetthe{proposition}
\newtheorem{definition}{Definition}

\newtheorem{observation}{Observation}

\newtheorem{remark}{Remark}
\newtheorem*{remark*}{Remark}
\newtheorem*{inftheorem*}{Theorem (informal)} 
\newtheorem*{notation*}{Notation}
\newtheorem*{observation*}{Observation}
\newtheorem*{theorem*}{Main theorem}
\newtheorem*{definition*}{Definition}
\newtheorem*{axiom*}{Axiom}
\newtheorem*{claim*}{Claim}
\newtheorem*{lemma*}{Lemma}

\usepackage{xspace}
\usepackage{xcolor}

\crefname{proposition}{Proposition}{Propositions}
\Crefname{proposition}{Proposition}{Propositions}
\crefname{fact}{Fact}{Facts}
\Crefname{fact}{Fact}{Facts}
\crefname{reduction}{Reduction}{Reductions}
\Crefname{reduction}{Reduction}{Reductions}
\crefname{observation}{Observation}{Observations}
\Crefname{observation}{Observation}{Observations}

\newcommand{\push}{\mathsf{push}}
\newcommand{\pull}{\mathsf{pull}}

\newcommand{\PNE}{\mathsf{PNE}}
\newcommand{\PNEexists}{\PNE^{?}}
\newcommand{\PNEfind}{\PNE^{\mathmagnifier}}

\newcommand{\NWR}{\mathsf{NWR}}
\newcommand{\NWRexists}{\NWR^{?}}
\newcommand{\NWRfind}{\NWR^{\mathmagnifier}}

\newcommand{\albe}{\mathsf{( \alpha,\beta )}}

\newcommand{\circuitSATexists}{\mathsf{CircuitSAT}^{?}}

\newcommand{\NUG}{\mathsf{NUG}}
\newcommand{\NUGexists}{\NUG^{?}}
\newcommand{\NUGfind}{\NUG^{\mathmagnifier }}
\newcommand{\NUGcount}{\NUG^{\#}}

\newcommand{\PNEcount}{\PNE^{\#}}
\newcommand{\NWRcount}{\NWR^{\#}}
\newcommand{\circuitSATcount}{\mathsf{CircuitSAT}^{\#}}

\newcommand{\Pclass}{\textsf{P}\xspace}
\newcommand{\RPclass}{\textsf{RP}\xspace}

\newcommand{\Pcount}{\textsf{\#P}\xspace}
\newcommand{\Pcounthard}{\Pcount-hard\xspace}
\newcommand{\Pcountcomplete}{\Pcount-complete\xspace}

\newcommand{\NP}{\textsf{NP}\xspace}
\newcommand{\NPhard}{\NP-hard\xspace}
\newcommand{\NPcomplete}{\NP-complete\xspace}

\newcommand{\PLS}{\textsf{PLS}\xspace}
\newcommand{\PLShard}{\PLS-hard\xspace}
\newcommand{\PLScomplete}{\PLS-complete\xspace}

\newcommand{\PPAD}{\textsf{PPAD}\xspace}

\newcommand{\PPADcomplete}{\PPAD-complete\xspace}

\newcommand{\hatm}{\hat{m}}
\newcommand{\Games}{\mathbf{G}}

\newcommand{\In}{\mathtt{In}}
\newcommand{\Out}{\mathtt{Out}}

\title{\Large{Avoiding the Worst: The Computational Complexity of \\ Not Worst-Responding}}

\author{%
\begin{tabular*}{0.95\textwidth}{@{}l@{\extracolsep{\fill}}r@{}}
\normalsize Mete \c{S}eref Ahunbay$^{\dagger}$              &  \\
\normalsize Paul W. Goldberg$^{\dagger}$                    & \\
\normalsize Edwin Lock$^{\ddagger}$                         &  \\
\normalsize Panayotis Mertikopoulos$^{\ast}$                & \small $^{\ast}$CNRS\\
\normalsize Bary S. R. Pradelski$^{\dagger \ast}$           &\small $^{\ddagger}$King's College London  \\
\normalsize Bassel Tarbush$^{\dagger}$                      & \small $^{\dagger}$University of Oxford
\end{tabular*}
}

\date{\normalsize \today}

\makeatletter
\renewenvironment{abstract}
  {\list{}{\leftmargin=0pt \rightmargin=0pt}\item\relax}
  {\endlist}
\makeatother

\begin{document}
\maketitle
\vspace{-1cm}
\onehalfspacing
\begin{abstract}
\noindent Finding, counting, or determining the existence of pure Nash equilibria, where players must play optimally given the others' actions, are known to be computationally intractable problems. 
We ask whether weakening optimality to the requirement that each player merely avoid worst responses yields tractable solution concepts. 
We show that it does not: any solution concept with this minimal guarantee is broadly ``as intractable'' as pure Nash equilibrium. In general games, determining the existence of \emph{no-worst-response} action profiles is \NPcomplete, and counting them is \Pcountcomplete. In potential games, where existence is guaranteed, the search problem is \PLScomplete.
However, a class of graphical games reveal a wedge in terms of complexity.
Computational intractability therefore stems not only from the requirement of optimality, but from the interactivity of games that makes even a minimal rationality guarantee \emph{for each player} intractable. Moreover, relaxing the latter requirement gives rise to a tractability trade-off between the strength of individual rationality guarantees and the fraction of players satisfying them.
\end{abstract}

\footnotebl{\emph{Keywords}: not worst responses, best responses, minimal rationality, satisficing equilibrium, pure Nash equilibrium, NP-completeness, PLS-completeness}

\footnotebl{\emph{Contact}: \texttt{mete.ahunbay@cs.ox.ac.uk}, \texttt{paul.goldberg@cs.ox.ac.uk},
\texttt{edwin.lock@kcl.ac.uk},\\
\texttt{panayotis.mertikopoulos@imag.fr},
\texttt{bary.pradelski@cnrs.fr},
\texttt{bassel.tarbush@economics.ox.ac.uk}
}

\footnotebl{\emph{Acknowledgements}: PWG was supported by UKRI/EPSRC grants EP/X040461/1 and EP/X038548/1. M\c{S}A, PM, and BSRP were supported by PEPR project FOUNDRY, ANR23-PEIA-0003.}

\section{Introduction}

A pure Nash equilibrium requires every player to choose a best response.
This is a highly demanding notion of individual optimality, and computing such outcomes is notoriously difficult:
the problem of deciding whether a (succinctly represented) finite game admits a pure Nash equilibrium, finding one, or counting their number is, respectively,
\NPcomplete,
\PLScomplete,
or
\Pcountcomplete 
\citep{gottlob2003pure,FabrikantPT04}.
These complexity classes are all believed to lie well beyond the capabilities of polynomial-time algorithms, so each of these problems is a fortiori considered intractable as a result.

From an operational standpoint, a natural response would be to study a relaxation of the problem, either at the level of the problem's search space (pure strategies) or its objective (unilateral optimality).
As an example of the former, if one allows for mixed strategies, the existence of a \emph{mixed} Nash equilibrium is guaranteed, but finding (the support of) an equilibrium remains computationally intractable\textemdash \PPADcomplete  \citep{daskalakis2009complexity}\textemdash 
and counting the number thereof may not be a well-posed problem altogether.
In this regard, the relaxation to mixed strategies does not render the problem of finding a Nash equilibrium more tractable.

In this paper, we consider a relaxation of pure Nash equilibria that maintains pure strategies and the restriction to ordinal comparisons as the object of interest, but relaxes the criterion of unilateral optimality.
To this end, suppose a player is not required to best-respond, but instead chooses one of their two best actions;
or one of their three best;
and so on.
The solution concept that results from every player ruling out all but their top $k$ actions is known in the literature as a \emph{$k$-satisficing equilibrium} \citep{pradelski2024satisficing}.
The top rung of the resulting satisficing hierarchy is the notion of a pure Nash equilibrium, in which every player best-responds to all other players.
By contrast, descending to the bottom rung, the weakest demand with any bite, requires no player choose their \emph{worst response}:
each player must merely avoid the \emph{single} action that does strictly worse than \emph{any} other alternative.

In view of this, a natural question that arises is whether there is an `optimality-tractability trade-off' along which tolerating `less optimal' behaviour results in a computationally easier problem.
Surprisingly, we find that reaching the bottom rung of the ladder is just as hard as climbing to the top:
\smallskip
\begin{quote}
\centering
\itshape
Any solution concept that guarantees that no player worst-responds to all other players
is as intractable as a pure Nash equilibrium.
\end{quote}
\smallskip
Hence, even a drastic relaxation of unilateral optimality does not result in any meaningful tractability dividends.
More specifically, determining the existence of an action profile where no player worst-responds, finding such a profile, or counting their number is computationally equivalent to the corresponding problems for pure Nash equilibria\textemdash \NPcomplete, \NPhard (\PLScomplete in potential games), and \Pcountcomplete, respectively.
(And, more generally, the same is true for every rung of the satisficing equilibrium ladder.)

There are several points to note in the above.
First, when each player has exactly two actions, avoiding a worst response is equivalent to best-responding, so in that case hardness is inherited from the corresponding pure Nash equilibrium problem.
The content of our results is that the equivalence survives \emph{beyond} two actions:
it holds for every fixed number $m \geq 2$ of actions per player, and the two concepts are genuinely distinct when $m > 2$.

Second, our hardness persists even in potential games (and identical-interest games), where a pure Nash equilibrium---and hence a no-worst-response action profile---is guaranteed to exist.

Third, the verdict is not absolute.
We identify a class of graphical games with many actions that provides a complexity-theoretic wedge, in which finding a no-worst-response action profile is easy, while finding a pure Nash equilibrium is hard.

Where, then, does a trade-off live---if at all?
As uniformly relaxing each player's requirement to play well is generally not sufficient to achieve a tractable solution concept, we study the requirement that only some fraction of players play well. If the guarantee is required only of a fraction $\alpha$ of players, and the guarantee itself is that a player plays within her top $\beta$ fraction of actions, then a sharp threshold emerges: the problem is tractable when $\alpha \leq \beta$ and hard when $\alpha > \beta$. Hence, to attain a tractable solution concept, both the requirement and the number of agents on which the requirement is imposed need to be relaxed.

\subsection{Our contributions}
 
We provide a complexity-theoretic characterisation of minimal individual rationality guarantees in succinctly represented finite games played with pure strategies. We take circuit games, in which the payoff functions are represented by Boolean circuits (see \Cref{sec:circuit-games} for details), as our model for succinct game representation. Throughout, we fix the number of actions per player to an arbitrary constant $m \geq 2$ and let the number of players grow.

\paragraph{General games.} In succinctly represented games with an arbitrarily fixed number of actions per player, deciding whether there exists an action profile in which no player plays a worst response is computationally as hard as deciding the existence of a pure Nash equilibrium (\Cref{thm:main}, \NPcomplete). The same equivalence holds for finding such an action profile (\Cref{cor:findNP}, \NPhard) and for counting the number of such action profiles (\Cref{cor:NumberHashP}, \Pcountcomplete).
 
\paragraph{Potential games.} We then turn to potential games, a class with many economic and technical applications. Although potential games are guaranteed to admit pure Nash equilibria \citep{rosenthal1973class,monderer1996potential}---and hence, a fortiori, no-worst-response action profiles---we show that \emph{finding} an action profile in which no player plays a worst response is computationally equivalent to finding a pure Nash equilibrium (\Cref{thm:potential-complexity}, \PLScomplete). A new reduction is required here (\Cref{remark:proof_technique} explains why existing constructions for the hardness of pure Nash equilibrium cannot be leveraged here), and its key ingredient is a family of circuit gadgets whose existence we establish using the probabilistic method and the Lov\'{a}sz Local Lemma. 
 
\paragraph{Graphical games.} Graphical games consist of games where a player's payoff depends only on the actions of her neighbours in an interaction graph. They have been widely studied in economics, operations and computer science, see, for example, \citet{kearns2001graphical,jackson2008social}. They are a natural model of strategic interaction among large populations, where each player typically interacts directly with only a few others. We show that a class of graphical games provide a complexity-theoretic wedge between no-worst-response action profiles and pure Nash equilibrium. While the latter remains computationally intractable (\Cref{prop:graphical-pne}, \NPcomplete), the former is guaranteed to exist (\Cref{thm:graphical-nwr}) and can be found in randomised polynomial time.

\paragraph{Non-universal guarantees.} Having established that ensuring guarantees for \emph{all} players is as hard as pure Nash equilibrium, we turn to guarantees required only of a fraction of players, and identify a sharp threshold separating tractability from hardness. A randomised polynomial-time algorithm finds an action profile in which a fraction at least $\alpha$ of players choose an action in their top-$\beta$ fraction of actions whenever $\alpha \leq \beta$ (\Cref{thm:threshold}), whereas finding such a profile is \NPhard when $\alpha > \beta$ (\Cref{thm:threshold}). Two special cases are worth recording. Finding a profile in which a fraction $\alpha$ of players best-respond is tractable only if $\alpha \leq 1/m$; finding one in which a fraction $\alpha$ of players do not worst-respond is tractable as soon as $\alpha \leq (m-1)/m$. Our negative results are in this sense knife-edge: with many actions per player, a profile in which almost nobody worst-responds is easy to find, while one in which even a small constant fraction of players best-respond is not.

\subsection{Related work}\label{sec:lit}

Computational complexity has been argued to provide a minimal plausibility requirement for economic solution concepts: predictions that rely on solving fundamentally hard problems may be descriptively or normatively untenable \citep{gilboa1989nash,camara}. In game theory, this perspective has called into question the plausibility of (pure) Nash equilibrium, as the concept has been shown to be computationally intractable in a variety of game classes (and our results show that any solution concept requiring all players to not worst-respond fails this fundamental plausibility requirement). Our work is related to the literature on computational complexity in games and in economic theory.
 
\paragraph{Computational complexity of equilibrium.}
Throughout the computational game theory literature, complexity results are
stated for games given in a \emph{succinct representation}, typically by
polynomial-size Boolean circuits that compute players' payoffs. This
convention rules out trivial hardness arising from exponentially large
payoff tables and allows computational hardness to be interpreted as a
genuine feasibility constraint rather than an artefact of input size (see
\Cref{sec:circuit-games,sec:compproblems}).
The computational study of equilibrium in games originates in the seminal
work of \citet{gilboa1989nash}, who showed that certain equilibrium-related
problems in finite games can be computationally intractable. The study of
the complexity of equilibrium notions has since flourished. For pure
strategies, the focus of our analysis, deciding whether a pure Nash
equilibrium exists in a general succinctly represented finite game is
\NPcomplete and the associated search problem is \NPhard
\citep{gottlob2003pure}.
Even in the class of potential games, where pure Nash equilibria are
guaranteed to exist \citep{rosenthal1973class,monderer1996potential},
finding a pure Nash equilibrium is \PLScomplete \citep{FabrikantPT04},
and best-response dynamics require exponential convergence time in the
worst case \citep{schaffer1991simple,anshelevich2008price}. Our results
show that action profiles in which no player worst-responds (in pure
strategies) are computationally equivalent to pure Nash equilibrium along
each of these dimensions.
Related work studies the approximate notion of $\epsilon$-equilibrium
\citep{radner1980collusive}. For pure strategies,
\citet{schoenebeck2012computational} establish several hardness results for
pure $\epsilon$-equilibria. While pure $\epsilon$-equilibrium and
no-worst-response action profiles are related, they are not equivalent and
neither implies the other. Taken together, computational hardness is a
pervasive feature of pure equilibrium concepts.
 
\paragraph{Complexity in economic theory.} Complexity considerations have a
long tradition in economic analysis. For example, computational models have
been studied in the context of repeated games
\citep{neyman1985bounded,rubinstein1986finite,kalai1988finite}, choice
\citep{salant2011procedural}, and learning \citep{wilson2014bounded}. Many
articles in this strand of the literature rely on finite automata.
\citet{camara} broke new ground by studying choice in a decision-theoretic
setting through the lens of computational complexity theory. Our work
contributes to the complexity-theoretic study of solution concepts in
games.

\section{Preliminaries}

We introduce our notation and recall standard notions from game theory and computational complexity theory.

\subsection{Games}
\label{sec:games}
 
A \emph{game} is a tuple
\[
\big(\,N, \,\{A_i\}_{i \in N}, \,\{u_i\}_{i \in N}\, \big)
\]
consisting of a set of players $N = \{1, \ldots, n\}$, a set of actions $A_i = \{0, \ldots, m_i-1\}$ for each player $i \in N$, as well as a payoff function $u_i: A \to \mathbb{R}$ for each $i \in N$, where $A := \times_{i \in N} A_i$ denotes the set of all \emph{action profiles}. In a standard abuse of notation, we write $(a_i,a_{-i}) \in A$ for the action profile in which player $i$ plays action $a_i \in A_i$ and $a_{-i} \in \times_{j \in N \setminus \{i\}} A_j$ denotes the actions played by the players in $N \setminus \{i\}$. Throughout, we work with games in which all players have the same number of actions, i.e., $m_i = m$; this is for expositional convenience, and it is without loss for our computational hardness results, which apply a fortiori to games with heterogeneous action sets.
 
Fix a player $i$, an action profile $a_{-i}$ of the other players, and an integer $k \in \{1, \ldots, m\}$. We say that $a_i \in A_i$ is a \emph{top-$k$ response} to $a_{-i}$ if fewer than $k$ of player $i$'s actions do strictly better against $a_{-i}$, that is, if
\[
    \bigl|\{\, a_i' \in A_i \;:\; u_i(a_i',a_{-i}) > u_i(a_i,a_{-i}) \,\}\bigr| < k .
\]
The two extremes of this scale are the main objects of our study. A top-$1$ response is a \emph{best response}: $a_i$ satisfies $u_i(a_i,a_{-i}) \geq u_i(a_i',a_{-i})$ for all $a_i' \in A_i$. At the other end, $a_i$ is a \emph{worst response} to $a_{-i}$ if $u_i(a_i,a_{-i}) < u_i(a_i',a_{-i})$ for all $a_i' \in A_i \setminus \{a_i\}$; equivalently, $a_i$ is a \emph{not-worst-response} precisely when it is a top-$(m-1)$ response. Observe that ties are permitted when best-responding, whereas worst-responding is strict: an action is a worst response only if every other action does strictly better. Not-worst-responding is therefore a genuinely weak demand---it is violated only by an action that is uniquely and strictly worst.
 
An action profile $a \in A$ is a \emph{pure Nash equilibrium} ($\PNE$) if, for each player $i \in N$, $a_i$ is a best response to $a_{-i}$, and it is a \emph{no-worst-response action profile} ($\NWR$ profile) if, for each $i \in N$, $a_i$ is a not-worst-response (or, equivalently, a top-$(m-1)$ response) to $a_{-i}$. Every $\PNE$ is an $\NWR$ profile, and when each player has two actions the two coincide, since with $m=2$ a top-$(m-1)$ response is a top-$1$ response. Note that an action profile in which each player plays a top-$k$ response is a $(k,\dots,k)$-satisficing equilibrium, in the terminology of \citet{pradelski2024satisficing}.

\subsection{Succinctly represented games}
\label{sec:circuit-games}
 
Before we can ask whether $\NWR$ profiles are hard to find, we must say how a game is presented to an algorithm. To not require the full payoff table as input, we consider a compact object from which payoffs can be computed on demand. We work within the formalism of \citet{feigenbaum1995game} and \citet{schoenebeck2012computational} who take circuit games as succinct representations of games. Circuit games are games in which each player's payoff function is presented as a Boolean circuit (see Appendix~\ref{sec:appendix} for a definition): the circuit takes as input a binary encoding of an action profile and outputs a binary encoding of that player's payoff.

\begin{definition}[Circuit games]
\label{def:circuit-games}
Fix $n,m \geq 2$ and let $s = \lceil\log_2 m\rceil$. A \emph{circuit game with $n$ players and $m$ actions per player} is a tuple
\[
    g = \bigl(n, C_1, \ldots, C_n\bigr),
\]
where each $C_i$ is a Boolean circuit with $ns$ input bits and two designated output blocks encoding, in unsigned binary, a non-negative integer $p_i$ and a positive integer $q_i$. The \emph{game induced by} $g$ is a tuple $(N, \{A_i\}_{i\in N}, \{u_i\}_{i \in N})$, where $N = \{1,\ldots,n\}$ is interpreted as the set of players, $A_i = \{0,\ldots,m-1\}$ is the action set for $i \in N$, and where, identifying each action profile $a=(a_1,\dots,a_n) \in A$ with the string in $\{0,1\}^{ns}$ that lists the $s$-bit binary encodings of $a_1,\dots,a_n$, the payoff of player $i$ at $a$ is the non-negative rational number
\[
    u_i(a) \;=\; \frac{p_i(a)}{q_i(a)} ,
\]
with $p_i(a)$ and $q_i(a)$ the integers output by $C_i$ on input $a$. 
\end{definition}

\begin{remark}
(i) Our assumption that all payoffs are non-negative is without loss of generality.\footnote{Adding a player-specific constant to all of a player's payoffs preserves all payoff comparisons and payoff differences, and therefore preserves all solution concepts and the potential-game property considered below.} (ii) Circuit games carry cardinal payoff information, but most of our results would be unchanged if payoffs were required to lie in $\{0,\ldots,m-1\}$, so that each circuit outputs $\lceil \log m \rceil$ bits in place of a rational number of unbounded precision. In other words, we could focus on purely ordinal information, with the only exception being \cref{sec:potential} on potential games.\footnote{There, we consider cardinal potential games (of which ordinal potential games are a superset), and while our \PLS-hardness result requires only ordinal information, our \PLS-membership result requires cardinal information.}
\end{remark}

For each fixed integer $m \geq 2$, we let $\Games_m$ denote the set of circuit games with $m$ actions per player, and we write
\[
    \mathbf{A}_m \;=\; \bigcup_{n \geq 2} \{0,\ldots,m-1\}^n
\]
for the set of all action profiles arising in these games.

\subsection{Computational problems}
\label{sec:compproblems}

We study three kinds of computational problem for $\NWR$ and $\PNE$, namely existence, search, and counting problems, denoted respectively by the superscripts $?$ and $\mathmagnifier$, and by the symbol $\#$. 
Fix a standard binary encoding of Boolean circuits---a list of gates, each recording its type and the indices of its inputs---and encode $g = (n,C_1,\ldots,C_n) \in \Games_m$ as the string $\langle g \rangle$ consisting of $n$ (in binary) followed by the binary encodings of $C_1,\ldots,C_n$. Similarly, encode an action profile $a = (a_1,\ldots,a_n) \in \mathbf{A}_m$ as the string $\langle a \rangle \in \{0,1\}^{ns}$ that concatenates the $s$-bit binary encodings of $a_1, \ldots, a_n$, as in \Cref{def:circuit-games}. Both encodings are injective, and we write
\[
    \mathbb B(\Games_m) = \bigl\{\, \langle g \rangle : g \in \Games_m \,\bigr\}
    \subseteq \{0,1\}^* ,
    \qquad
    \mathbb B(\mathbf{A}_m) = \bigl\{\, \langle a \rangle : a \in \mathbf{A}_m \,\bigr\}
    \subseteq \{0,1\}^*
\]
for their images, where $\{0,1\}^*$ denotes the set of all binary strings. The size of an instance is $|\langle g\rangle|$. We suppress the brackets wherever no confusion can arise, writing $g$ for $\langle g\rangle$ and $a$ for $\langle a\rangle$. The advantage of working with circuit games is that there need not be an exponential blow-up in $|\langle g\rangle|$ as the number of agents increases.

\begin{definition}[Decision problems: $m$-$\NWRexists$ and $m$-$\PNEexists$]
\label{def:decide}
The problem of deciding whether a game $g \in \Games_m$ admits a no-worst-response action profile is the relation
\[
    m\text{-}\NWRexists \subseteq \mathbb B(\Games_m) \times \{0,1\},
\]
where $(g,1) \in m\text{-}\NWRexists$ if and only if $g$ admits a no-worst-response action profile, and $(g,0) \in m\text{-}\NWRexists$ otherwise. The problem $m$-$\PNEexists$ is defined analogously.
\end{definition}
 
\begin{definition}[Search problems: $m$-$\NWRfind$ and $m$-$\PNEfind$]
\label{def:find}
The problem of finding a no-worst-response action profile in a game $g \in \Games_m$ is the relation
\[
    m\text{-}\NWRfind \subseteq \mathbb B(\Games_m)  \times \mathbb B(\mathbf{A}_m),
\]
where $(g,a) \in m\text{-}\NWRfind$ if and only if $a$ is a no-worst-response action profile of $g$. An algorithm solves $m$-$\NWRfind$ if, on input $g$, it outputs a no-worst-response action profile of $g$ when one exists, and returns $\bot$ when none does. The problem $m$-$\PNEfind$ is defined analogously.
\end{definition}

\begin{definition}[Counting problems: $m$-$\NWRcount$ and $m$-$\PNEcount$]
\label{def:count}
The problem of counting the number of no-worst-response action profiles is
the relation
\[
    m\text{-}\NWRcount \subseteq \mathbb B(\Games_m)  \times \{0,1\}^* ,
\]
where $(g,k) \in m\text{-}\NWRcount$ if and only if $k$ is the binary
representation of the number of no-worst-response action profiles of $g$. The problem $m$-$\PNEcount$
is defined analogously.
\end{definition}

\subsection{Computational complexity}
\label{sec:compcomplexity}
 
Our computational complexity results are, fundamentally, comparative. Each of them says that some problem about no-worst-response profiles is exactly ``as hard'' as the corresponding problem about pure Nash equilibrium. Comparing computational hardness is done via reduction. We provide the formal definitions of polynomial-time reductions and of our complexity classes of interest, including \Pclass, \NP, and \PLS, in Appendix~\ref{sec:appendix}.

\section{The complexity of not-worst-responding}\label{sec:NWRcomplexity}

We now analyse the computational tractability of no-worst-response action profiles in general, potential, and graphical games.

\subsection{The complexity of not-worst-responding in general games}

We begin our investigation of the intractability of not-worst-responding by studying the problem of determining whether a given circuit game in $\Games_m$ has a no-worst-response action profile. Here, and throughout, we show hardness for the computational problems in which the number of actions for all players is fixed to some integer $m \geq 2$. This also implies hardness for the version of the problems in which the number of actions is part of the input.

\begin{theorem}\label{thm:main}
For every integer $m \geq 2$, the decision problem $m$-$\NWRexists$ of determining the existence of a no-worst-response action profile, restricted to games with $m$ actions per player, is \textnormal{\NPcomplete}. 
\end{theorem}

It is easy to check that $m$-$\NWRexists \in \NP$, as we can guess a no-worst-response action profile and verify in polynomial time that no player is worst-responding.%
\footnote{We note that membership in \NP also holds when $m$ is not fixed: here the membership certificate consists of a guessed $\NWR$ profile $a$ together with another action $a'_i$ for each agent that certifies the $\NWR$ property by giving a weakly lower payoff than $a_i$ (given other actions $a_{-i}$).}
To show that $m$-$\NWRexists$ is \NPhard, we define a reduction from $\circuitSATexists$, which is known to be \NPcomplete, and then show that it is polynomial. 
Note that our proof implies that no-worst-response action profiles are not guaranteed to exist (also previously shown in \citealp{pradelski2024satisficing}).

\begin{reduction}[Reduction from $\circuitSATexists$ to {$m$-$\NWRexists$}]
\label{def:circuitSAT-reduction}
Let $C$ be a Boolean circuit with inputs $x_1, \ldots, x_n$, and fix some $m \geq 2$. We construct a game $g$ from $C$ with $n m$ players $(i,\ell)$, where $i \in \{1, \ldots, n\}$ and $\ell \in \{1, \ldots,m\}$. Each player has actions $A_{i\ell} = \{0, \ldots, m-1\}$. 

To construct the payoff function of each player, we first define a function $f : A \to \{0,1\}$ as follows: let $f(a) = 1$ if, for each $i \in \{1,\ldots,n\}$, the players $(i,1), \ldots, (i, m)$ play the same action $a_{i\ell} \in \{0,1\}$ and $C(\tau) = 1$ for $\tau$ defined by $\tau_i = a_{i\ell}$. Otherwise, let $f(a) = 0$.

This allows us to define payoffs as
\begin{equation}
    u_{i\ell}(a) =
    \underbrace{(m+1) f(a)}_{\text{global term}} +
    \underbrace{ \left [ \left (\ell + \sum_{\ell'=1}^{m} a_{i\ell'} \right ) \mkern-10mu \mod m \right ]}_{\text{local term}} ,
    \label{eq:payoffs}
\end{equation}
which completes the description of the reduction. \hfill$\triangle$
\end{reduction}

Intuitively, the global term reflects whether the actions map to a satisfying truth assignment of the circuit $C$. Meanwhile, for each group $(i, \ell)_{\ell \in \{1, \ldots, m\}}$, the local term enforces cyclic behaviour amongst players as long as the truth assignment is not satisfied.
This leads to the following straightforward observation.
\begin{observation}
\label{obs:circuitSAT-reduction}
Each player's payoff from action profile $a$ is at least $m+1$ if $f(a) = 1$ and at most $m$ if $f(a) = 0$. This implies the following for any action profile $a$.
\begin{enumerate}[leftmargin=*, labelindent=0pt,label=(\roman*)]
    \item If $f(a)=1$, then each player gets a payoff of at least $m+1$ from the global term. By unilaterally deviating to, say, action $a'_{i\ell}$, the player forces $f(a'_{i\ell}, a_{-i\ell}) = 0$, so her payoff strictly reduces.
    \item If $f(a)=0$, then for each group $(i, \ell)_{\ell \in \{1, \ldots, m\}}$, one player has a payoff of $0$, as the global term in \eqref{eq:payoffs} is $0$ and the local term as a function of $\ell$ is a bijection from $\{1,\ldots,m\}$ to $\{0,\ldots,m-1\}$. Any such player with payoff $0$ can profitably deviate by playing any other action in $A_{i\ell} \setminus \{a_{i\ell}\}$, so she is worst-responding at $a$.
\end{enumerate}
\end{observation}

\begin{proof}[Proof of \Cref{thm:main}]
Let $C$ be a Boolean circuit with inputs $(x_1, \ldots, x_n)$, and let $g$ be the game as constructed above. Suppose circuit $C$ is unsatisfiable. By construction of $f$, we have $f(a) = 0$ for any action profile, and so \cref{obs:circuitSAT-reduction} tells us that every action profile has some players who worst-respond.

Now suppose that $C$ is satisfiable, and let $\tau^*$ be a satisfying truth assignment for $C$. Define the action profile $a$ given by $a_{i\ell} = \tau^*_i$ for each $i \in \{1,\ldots,n\}$ and $\ell \in \{1,\ldots,m\}$. It is immediate by construction that $f(a) = 1$, so by \cref{obs:circuitSAT-reduction} we see that each player gets a payoff of at least $m+1$. Moreover, unilaterally deviating strictly reduces a player's payoff, so $a$ is a pure Nash equilibrium (and so also a no-worst-response action profile).

In summary, $(g,1) \in m$-$\NWRexists$ if and only if $(C,1) \in \circuitSATexists$. Finally, it is straightforward to construct Boolean circuits to compute each $u_{i\ell}$ in polynomial time, so we arrive at a polynomial-time reduction.
\end{proof}

Since $m$-$\NWRfind$ cannot be easier than $m$-$\NWRexists$, we immediately obtain the following corollary.

\begin{corollary}\label{cor:findNP}
For every integer $m \geq 2$, the search problem $m$-$\NWRfind$ of finding a no-worst-response action profile, restricted to games with $m$ actions per player, is \textnormal{\NPhard}.
\end{corollary}

The game $g$ constructed from circuit $C$ has $k$ no-worst-response action profiles (in particular, $k$ pure Nash equilibria) exactly when there are $k$ satisfying assignments for $C$. Moreover, it is easy to check that $m$-$\NWRcount \in \Pcount$. And, since $\circuitSATcount$ is \Pcountcomplete, the same holds for $m$-$\NWRcount$.

\begin{corollary}\label{cor:NumberHashP}
For every integer $m \geq 2$, the counting problem $m$-$\NWRcount$ of determining the number of no-worst-response action profiles, restricted to games with $m$ actions per player, is \textnormal{\Pcountcomplete}.
\end{corollary}
\begin{proof}[Proof of \Cref{cor:NumberHashP}]
In the reduction from circuit $C$ to game $g$, we see that all action profiles $a$ for $g$ have a worst-responding player if $f(a) = 0$. This is the case in particular if $C$ is not satisfiable. Moreover, when $C$ is satisfiable, $a$ is a no-worst-response action profile if and only if $f(a) = 1$. It is straightforward that these correspond one-to-one to satisfying truth assignments for $C$, so $m$-$\NWRcount$ is \Pcounthard.
\end{proof}

\subsection{The complexity of not-worst-responding in potential games}\label{sec:potential}

We now turn our attention to potential games. It is well-known that every potential game is guaranteed to have a pure Nash equilibrium \citep{rosenthal1973class,monderer1996potential}. Moreover, the search problem of finding a pure Nash equilibrium in such games is known to be \PLScomplete \citep{FabrikantPT04}. Our main result in this section will be that finding no-worst-response action profiles in potential games has the same complexity.

\begin{definition}[Potential games]\label{def:potential-game}
A game $(N, \{A_i\}_{i \in N}, \{u_i\}_{i \in N})$ is a \emph{potential game} if there exists a function $\phi:A \to \mathbb{R}$ such that for each $a \in A$, each $i \in N$, and each $a_i' \in A_i$, we have $u_i(a_i',a_{-i}) - u_i(a) = \phi(a_i',a_{-i}) - \phi(a)$.    
\end{definition}

Throughout this section the input to our computational problems is always a circuit game $g$, as defined in \Cref{sec:circuit-games}, together with the guarantee that the game it induces is a potential game. The potential itself is not supplied (\Cref{lem:PLSmembership} below reconstructs it from the payoff circuits).

We first remark that finding a pure Nash equilibrium in potential games is \PLScomplete even if each player has two actions. This allows us to reduce from this restricted problem ($2$-$\PNEfind$), to $m$-$\NWRfind$ on the class of potential games.

\begin{fact}[following from \citealp{pardalos1992complexity}]
\label{fact:2-action-pot-games-hard}
The search problem $2$-$\PNEfind$ of finding a pure Nash equilibrium, restricted to potential games with two actions per player, is \PLScomplete.
\end{fact}

For completeness, we discuss how this fact follows from \citet{pardalos1992complexity}. First, it is immediate that $2$-$\PNEfind$ lies in \PLS. To show hardness, we reduce from the local binary quadratic program problem. Recall that a binary quadratic program is specified by an $n \times n$ matrix $Q$ and a length $n$ vector $q$, both integral. The problem asks to find an optimal solution of
\begin{equation}\label{opt:BQP}
    \max_{x \in \{0,1\}^n} q^\top x + \frac{1}{2} x^\top Qx .
\end{equation}
The local variant, in turn, asks to find a $0$-$1$ vector $x$ such that no change in a single coordinate improves the value of the solution. \citet{pardalos1992complexity} show that this search problem is \PLScomplete. This local optimisation variant admits an obvious representation as a potential game with $n$ players where each player $i \in N = \{1,\ldots,n\}$ has action set $A_i = \{0,1\}$, and payoff functions are given by
\[
u_i(a) = \sum_{j = 1}^n \left[q_ja_j + \sum_{k =1}^n\frac{1}{2} Q_{jk}a_ja_k \right].
\]
Then the set of pure Nash equilibria of the game is in bijection with the set of local optima of \eqref{opt:BQP}.

We are now ready to develop our main result in this section.

\begin{theorem}\label{thm:potential-complexity}
For every integer $m \geq 2$, the search problem $m$-$\NWRfind$ of finding a no-worst-response action profile, restricted to potential games with $m$ actions per player, is \textnormal{\PLScomplete}.
\end{theorem}

\begin{remark}\label{remark:proof_technique}
Before turning to the proof, we note that existing constructions for the 
hardness of pure Nash equilibrium cannot be used to prove \cref{thm:potential-complexity}  and a new gadget is in fact required. Known \PLScomplete~problems fall into three classes: 
(i) optimisation problems without a simple multi-agent-game interpretation (e.g., graph partitioning \citep{schaffer1991simple}), or (ii) games with  exactly two actions per player (e.g., local MAXCUT, MAX2SAT \citep{schaffer1991simple}), or (iii) games with ``many'' actions per player (e.g., constraint assignment \citep{dumrauf2013pls}, congestion games \citep{FabrikantPT04}).

Reducing to (i) is not useful for proving hardness of $\NWR$, and our result is only obvious in case (ii), where $\PNE$ and $\NWR$ action profiles coincide; otherwise, a wedge in complexity is possible. Our gadget allows reducing case (ii) to (iii) while maintaining the $\PNE$-$\NWR$ equivalence on action profiles. Since $\NWR$ excludes only strict worst responses, naïve constructions such as simply cloning actions fail. For example, if each player has two copies of actions $\{0,1\}$, then no action is a strict worst response, because for every action there exists another with the same utility.
\end{remark}

\Cref{thm:potential-complexity} follows from a series of lemmas. \Cref{lem:PLSmembership} below establishes \PLS membership. The remaining arguments and results in this section establish \PLS-hardness.

\begin{lemma}\label{lem:PLSmembership}
    For every integer $m \geq 2$, the search problem $m$-$\NWRfind$ of finding a no-worst-response action profile, restricted to potential games with $m$ actions per player, is in \PLS.
\end{lemma}
\begin{proof}[Proof of \Cref{lem:PLSmembership}]
Suppose we have a potential game with a list of players $1,2,\ldots,n$, with $m$ actions per player, and a potential function $\phi$. Firstly, note that the number of actions for each player is constant and in our model it is straightforward to generate an initial action profile $a^*$. Secondly, since $\phi$ is a potential function, it is defined only up to an additive constant, and we see that we may evaluate it up to this constant term.

Indeed, let $a'$ be any action profile. For $0 \leq k\leq n$, let $a^k\in A$ be defined by setting $a^k_i=a'_i$ if $i \leq k$ and $a^k_i=a^*_i$ otherwise, so that $a^0=a^*$ and $a^n=a'$, while $a^{k-1}$ and $a^k$ differ only in player $k$'s action. By the potential game property applied to each such single-player deviation,
\[
    \phi(a')-\phi(a^*) = \sum_{k=1}^n\bigl(\phi(a^k)-\phi(a^{k-1})\bigr) = \sum_{k\in N}\bigl(u_k(a^k)-u_k(a^{k-1})\bigr).
\]
Therefore, the potential is computable at any action profile $a'$, up to a constant which renders $\phi(a^*)=0$.

Since all utilities are rational numbers whose numerators and denominators have bit representations of length bounded by the size of the corresponding circuit, the numerator and denominator of the above sum have bit length $L=O(n|C|)$, where $|C|=\max_i|C_i|$. In particular, after absorbing fixed encoding constants into $|C|$, we may take $L=n|C|$. Moreover, for every action profile $a'$,
\[
    \left|\phi(a')-\phi(a^*)\right|\leq n2^{|C|}.
\]
Let $B=n2^{|C|}$ and define the auxiliary potential 
\[
\hat{\phi}(a)=
\begin{cases}
    \phi(a)-\phi(a^*), & \text{if some player is worst-responding at $a$},\\
    B+1, & \text{otherwise}.
\end{cases}
\]
Thus, every no-worst-response profile has strictly greater auxiliary potential than every profile with a worst-responding player. Moreover, at any profile with a worst-responding player, changing that player's action gives a strict increase in $\hat{\phi}$: either the resulting profile has no worst-responding player, or the potential game property implies a strict increase in $\phi$.

It remains to make the objective integer-valued. The numerator and denominator of every value in the image of $\hat{\phi}$ have at most $L$ bits each. Hence, if $x\neq y$ are two such values, then
\[
    |x-y|\geq \frac{1}{2^{2L}}.
\]
Shift $\hat{\phi}$ to make it non-negative, and round it to the grid of width $2^{-(2L+1)}$. Formally, let
\[
    \phi'(a) = \operatorname{round}\!\left(2^{2L+1}\bigl(\hat{\phi}(a)+B\bigr)\right) \in\mathbb{N}.
\]
Equivalently, $\hat{\phi}(a)+B$ is rounded to the nearest fraction $k/2^{2L+1}$ and $\phi'(a)$ returns $k$. Since the rounding error in each value is at most $2^{-(2L+2)}$, the combined error when comparing two values is at most $2^{-(2L+1)}$, strictly less than their minimum separation. Consequently,
\[
    \hat{\phi}(a)<\hat{\phi}(a') \quad\Longleftrightarrow\quad \phi'(a)<\phi'(a').
\]
Thus $\phi'$ has exactly the same local improvements and local optima as $\hat{\phi}$.

Finally, $\phi'$ has $O(n|C|)$ bits and can be evaluated in polynomial time. Checking whether a player is worst-responding and finding an improving neighbour require only polynomially many circuit evaluations, because $m$ is fixed. The local optima of $\phi'$ are precisely the no-worst-response action profiles, so the search problem belongs to \PLS.
\end{proof}

We now turn to showing \PLS-hardness. We note that the reduction from $\circuitSATexists$ we invoked in the proof of \cref{thm:main} is insufficient to infer \cref{thm:potential-complexity}; in particular, the local term of \eqref{eq:payoffs} is cyclic, so the reduced game is not a potential game. Thus, to prove hardness, our first step is to develop a reduction from 2-$\PNEfind$, restricted to potential games with two actions per player, to $\hat m$-$\NWRfind$, restricted to potential games with $\hat m$ actions per player, that preserves the potential property of the game.

\begin{reduction}[Reduction from 2-$\PNEfind$ to $\hat m$-$\NWRfind$]
\label{def:potential-reduction}
Fix integers $\hatm \geq 2$ and $q = 6 \hatm 2^{\hatm}$. Suppose we are given $g \in \Games_2$ whose induced game is a potential game, with players $N=\{1,\dots,n\}$, action sets $A_i = \{0,1\}$ and payoffs $u_i$.%
\footnote{We note here that the reduction, and the subsequent proofs, extend immediately to any fixed number $m$ of actions for the players in the game $g$, but for our purposes $m=2$ suffices.} By \cref{fact:2-action-pot-games-hard}, the search problem for this class of circuit games is \PLShard.

From this, we now construct a game $\hat{g} = (\,\hat N, \,\{\hat A_i\}_{i \in \hat N}, \,\{ \hat u_i\}_{i \in \hat N}\,) \in \Games_{\hatm}$, with $\hatm$ actions per player as follows. The set of players
\[
\hat N \coloneqq \{(i,\ell) : i \in N \text{ and } \ell \in \{1, \dots, q\}\}
\]
consists of $n$ groups $[i] \coloneqq \{(i,\ell) : \ell \in \{1,\dots,q\}\}$, one for each $i \in N$. Each player $(i, \ell)$ has action set $\hat A_{i\ell} \coloneqq \{0,\dots, \hatm -1\}$. For any group $[i]$, define $a_{[i]}$ as the restriction of action profile $a$ to the group's players, and let $\hat{A}_{[i]} = \{0, \ldots, \hatm-1\}^q$ be the group's action space. To define the payoff functions, we first introduce a \textit{circuit gadget} $\rho_i: \hat{A}_{[i]} \to A_{i}$ for each $i \in N$. This gadget aggregates the actions $\hat{a}_{[i]}$ of a group in the game $\hat{g}$ into a single action $\rho_i(\hat{a}_{[i]})$ of player $i$ in the game $g$. By applying the gadget to all groups' actions, we map an action profile $\hat{a}$ in the new game $\hat{g}$ back to an action profile $\rho(\hat{a})$ in the original game $g$ by defining $\rho: \hat{A} \to A$ as
\[
    \rho(\hat{a}) = (\rho_i(\hat{a}_{[i]}))_{i \in N}.
\]
We now define the payoff of each player $(i,\ell)$ in game $\hat{g}$ as
\begin{equation}
\label{eq:potential-payoffs}
    \hat{u}_{i\ell}(\hat{a}) = u_i(\rho(\hat{a})).
\end{equation}
Crucially, we require that each gadget $\rho_i$ satisfies the following property ($\star$): at each action profile in $\hat A_{[i]}$ among players in the group $[i]$, and for each action $a_i$ not currently taken by player $i$ in $g$, there is some player in the group for whom any action change must change the action of player $i$ to $a_i$. Formally:
\begin{itemize}[leftmargin=*, labelindent=10pt]
  \item[($\star$)] For each $i\in N$, each $\hat a_{[i]} = (\hat a_{i\ell}, \hat a_{-i\ell}) \in \hat A_{[i]}$, and each $a_i \in A_i \setminus \{\rho_i(\hat a_{[i]})\} $, there is some $\ell \in \{1, \dots, q\}$ such that $\rho_i(\hat a'_{i\ell}, \hat a_{-i\ell} ) = a_i$ holds for each $\hat a'_{i\ell} \neq \hat a_{i\ell}$.\label{item:rho-property}
\end{itemize}

Intuitively, in the game, we replace each player $i$ with a coalition of players $[i]$, so that the actions $\hat a_{[i]}$ of players in $[i]$ together determine the action $a_i$ taken by player $i$ through some choice function $\rho_{i}$. The property ($\star$) is a sensitivity condition; for every deviation $a_i \rightarrow a'_i$ of player $i$, there must exist some player $\ell \in [i]$ in the coalition such that \emph{every} deviation of $\ell$'s action results in player $i$ deviating to $a'_i$. Consequently, if $a'_i$ is a utility-increasing deviation for player $i$, $\ell$ must be worst-responding.

Setting $m=2$ in \Cref{lem:gadget} below establishes the existence of circuit gadgets satisfying ($\star$). \hfill$\triangle$
\end{reduction}

\begin{lemma}
\label{lem:gadget}
Fix integers $m, \hatm \geq 2$, and $q \geq 16 m^{\hatm-1} \hatm \ln m$. There exists a function 
$
f : \{0,\dots,\hatm-1\}^q \to \{0,\dots,m-1\}
$
satisfying the following: for any $x \in \{0,\dots,\hatm-1\}^q$ and any $k \in \{0,\dots,m-1\} \setminus \{f(x)\}$, there is some $\ell \in \{1,\dots,q\}$ such that $f(x_\ell' , x_{-\ell}) = k$ for each $x_\ell' \neq x_\ell$.
\end{lemma}

\begin{proof}[Proof of \Cref{lem:gadget}]
We start by arguing that the problem is identical to the existence of a certain grid colouring. For any vertex $x$ in a grid $G:=\{0,\dots,\hatm-1\}^q$ and any $\ell \in \{1,\dots,q\}$, we refer to
\[
L_\ell(x) := \{(z,x_{-\ell}) : z \in \{0,\dots,\hatm-1\}\}
\]
as the \emph{line} going through $x$ in the $\ell$th direction. For each vertex $x$ in the grid, we assign a \emph{colour} $f(x) \in \{0,\dots,m-1\}$. We'll call a vertex $x \in G$ \emph{bad} if there is some $k \neq f(x)$ such that, for every direction $\ell$, there is some vertex $y \in L_\ell(x) \setminus \{x\}$ for which $f(y) \neq k$. Proving the lemma is therefore equivalent to showing that there is some value of $q$ (which we determine below) and a function $f$ that assigns colours from $\{0,\dots,m-1\}$ to vertices in the grid $G$ such that no vertex is bad.

We now show that such a colouring exists via the probabilistic method, using the Lov\'{a}sz Local Lemma (LLL); see \citet{alon-spencer}. For each vertex in the grid $G$, independently assign a colour uniformly at random from $\{0,\dots,m-1\}$. For each $x \in G$, let $B_x$ denote the event that $x$ is bad. For the events $\{B_x\}_{x \in G}$, the LLL states that if (i) $\mathbb{P}[ B_x ] \leq p$ for each $x \in G$ and (ii) each $B_x$ is independent of all the other events except for at most $d$ of them, then 
\[
e(d+1)p \leq 1
\]
implies that there is a non-zero probability that none of the vertices are bad. We next determine values for $p$ and $d$ (in terms of $q$) in our problem. It is then sufficient to choose a value of $q$ that is large enough such that the LLL condition holds.

\begin{enumerate}[leftmargin=*, labelindent=0pt,label=(\roman*)]
\item Fix a vertex $x \in G$ and a colour $k \in \{0,\ldots,m-1\}$. Fixing some direction $\ell \in \{1,\ldots,q\}$, the probability that every $y \in L_\ell(x) \setminus \{x\}$ has colour $k$ is $\lambda \coloneqq (\tfrac{1}{m})^{\hatm - 1}$. The probability that every direction has some vertex whose colour is not $k$ is therefore $(1 - \lambda)^q$. By a union bound over all colours $k \neq f(x)$, we obtain that
\[
\mathbb{P}[B_x] \leq (m-1) (1 - \lambda )^q \leq m e^{-\lambda q} \eqqcolon p.
\]
\item The event $B_x$ depends only on the colours of vertices lying on a line through $x$. Therefore, the events $B_x$ and $B_{x'}$ are independent whenever $x$ and $x'$ differ in at least three entries (because then there is no vertex $y$ that lies on a line through both $x$ and $x'$). It follows that the event $B_x$ is mutually independent of all but at most $d$ other events, where
\[
d \leq 1 + q(\hatm -1) + \binom{q}{2}(\hatm - 1)^2 \leq q^2 \hatm^2 - 1.
\]
The term 1 counts $x$ itself, the term $q(\hatm- 1)$ counts all vertices at a Hamming distance of 1 from $x$, and the term $\binom{q}{2}(\hatm - 1)^2$ counts all vertices at a Hamming distance of 2 from $x$.
\end{enumerate}

Using the upper bounds from arguments (i) and (ii) above, the LLL condition becomes $e q^2 \hatm^2 m e^{-\lambda q} \leq 1$. Rearranging gives us
\[
\frac{q^2}{e^{\lambda q}} \leq \frac{1}{e \hatm^2 m} .
\]
Since $x^2 \leq e^x$ for all $x >0$, we can set $x = \tfrac{\lambda}{2}q$ to conclude that $\tfrac{q^2}{e^{\lambda q}} \leq \tfrac{4}{\lambda^2} e^{-\tfrac{\lambda}{2}q}$. For the LLL condition to hold, it is therefore sufficient to have
\[
\frac{4}{\lambda^2} e^{-\tfrac{\lambda}{2}q} \leq \frac{1}{e \hatm^2 m},
\]
which re-arranges to
\[
q \geq \frac{2}{\lambda}\left(\ln(4 e \hatm^2 m) - 2\ln\lambda\right) =  2 m^{\hatm-1} \left(  1 + \ln(4) + 2\ln(\hatm) + (2 \hatm - 1)\ln(m) \right) .
\]
Upper bounding each element of the sum in parentheses by $2 \hatm \ln(m)$ gives a crude upper bound of $\left(  1 + \ln(4) + 2\ln(\hatm) + (2 \hatm - 1)\ln(m) \right) \leq 8 \hatm \ln(m)$. For the LLL condition to hold, it is therefore sufficient that
\[
q \geq 16 m^{\hatm-1} \hatm \ln(m) ,
\]
and this completes the proof. 
\end{proof}

We now prove that \cref{def:potential-reduction} is indeed a reduction (in the sense of \cref{def:reduction} given in Appendix~\ref{sec:appendix}), in a series of lemmas.\footnote{Let $\Phi_m$ denote the set of all $g \in \Games_m$ such that the game induced by $g$ is a potential game. Then all computational problems considered in \cref{sec:potential} have instance set $\mathbb B(\Phi_m)$ and output sets $\{0,1\}$ (for decision problems) or $\mathbb B(\mathbf A_m)$ (for search problems). \cref{def:potential-reduction} is then a reduction with $\In=\mathbb B(\Phi_2)$, $\In' = \mathbb B(\Phi_{\hat m})$, $\push(g) = \hat g$, and $\pull(g, \hat a) = \rho(\hat a)$. Then \cref{lemma:g-hat-is-potential-game} establishes that $\push$ does indeed land in $\In'$. \cref{lemma:pls-reduction-property} establishes condition (i) and \cref{lemma:game-admits-nwr-profile} establishes condition (ii) of \cref{def:reduction}.}

\begin{lemma}
\label{lemma:g-hat-is-potential-game}
If game $g$ is a potential game, then $\hat{g}$ is a potential game.
\end{lemma}
\begin{proof}[Proof of \Cref{lemma:g-hat-is-potential-game}]
Suppose game $g$ is a potential game, and let $\phi$ be its potential function. Fix a player $(i,\ell)$ and an action profile $\hat{a} = (\hat{a}_{i\ell}, \hat{a}_{-i\ell})$ in the game $\hat{g}$. Now, let $\hat{a}' = (\hat{a}'_{i\ell}, \hat{a}_{-i\ell})$ be an action profile in which the player unilaterally deviates from $\hat{a}$. We also define $a = \rho(\hat{a})$ and $a' = \rho(\hat{a}')$. By construction of $\rho$, we see that $a' = (\rho_i(\hat{a}'_{[i]}), a_{-i})$.  Applying the payoff definition \eqref{eq:potential-payoffs}, we see that
\begin{align*}
\hat u_{i\ell}(\hat{a}'_{i\ell}, \hat{a}_{-i\ell}) - \hat u_{i\ell}(\hat{a})
    &= u_i(\rho(\hat{a}'_{i\ell}, \hat{a}_{-i\ell})) - u_i(\rho(\hat{a})) \\
    &= u_i(a'_i, a_{-i}) - u_i(a) \\
    &= \phi(a'_i, a_{-i}) - \phi(a) \\
    &= \phi(\rho(\hat{a}'_{i\ell}, \hat{a}_{-i\ell})) - \phi(\rho(\hat{a})).
\end{align*}
So $\hat{\phi}$ defined by $\hat{\phi}(\hat{a}) = \phi(\rho(\hat{a}))$ for any action profile $\hat{a}$ of game $\hat{g}$ is a potential function for game $\hat{g}$, proving that $\hat{g}$ is a potential game.
\end{proof}

Every potential game has a pure strategy Nash equilibrium, which is necessarily a local maximum of the potential function \citep{monderer1996potential}. Since any pure Nash equilibrium is necessarily a no-worst-response action profile, we infer the existence of the latter:
\begin{corollary}
\label{lemma:game-admits-nwr-profile}
Game $\hat{g}$ admits a no-worst-response action profile.
\end{corollary}

More interestingly, our reduction ensures that no-worst-response action profiles stand in correspondence with the pure Nash equilibria of the original game.

\begin{lemma}
\label{lemma:pls-reduction-property}
Fix some action profile $\hat{a}$ of game $\hat{g}$ and denote by $a = \rho(\hat{a})$ the corresponding action profile of game $g$. Then $\hat{a}$ is a no-worst-response profile of $\hat{g}$ if and only if $a$ is a pure Nash equilibrium of $g$.
\end{lemma}
\begin{proof}[Proof of \Cref{lemma:pls-reduction-property}]
Suppose $a$ is a pure Nash equilibrium of $g$. We claim that $\hat{a}$ is a pure Nash equilibrium of $\hat{g}$, and hence in particular is a no-worst-response action profile. Let $(i,\ell) \in \hat N$ be any player in $\hat{g}$, and consider any unilateral deviation from $\hat{a}$ to $\hat{a}' := (\hat{a}'_{i\ell}, \hat a_{-i\ell})$ by that player. By construction of $\rho$, we see that $\rho(\hat a')$ and $\rho(\hat a) = a$ can differ only in the $i$th component. Assuming they differ, there must exist some $a_i' \in A_i$ such that $\rho(\hat a') = (a_i',a_{-i})$. Since $a$ is a pure Nash equilibrium in $g$, no such deviation is profitable, so $u_i(\rho(\hat a')) \le u_i(\rho(\hat a)) = u_i(a)$. But by construction of the utilities in $\hat g$, this implies that 
\[
  \hat u_{i\ell}(\hat a') = u_i(\rho(\hat a')) \le u_i(\rho(\hat a)) = \hat u_{i\ell}(\hat a).
\]
So $\hat a$ is a pure Nash equilibrium, and thus a no-worst-response action profile, of $\hat g$.

Conversely, suppose that $a = \rho(\hat a) \in A$ is not a pure Nash equilibrium of $g$. Then there exists some player $i \in N$ and some action $a_i' \in A_i$ such that $u_i(a_i',a_{-i}) > u_i(a) = u_i(a_i,a_{-i})$. By property ($\star$), there is a player $(i,\ell)$ in group $[i]$ of the game $\hat g$ for whom playing any $\hat a'_{i\ell} \neq \hat a_{i\ell}$ yields
\[
  \hat u_{i\ell}(\hat a') = u_i(\rho(\hat a')) = u_i(a_i',a_{-i}) > u_i(a) = u_i(\rho(\hat a)) = \hat u_{i\ell}(\hat a).
\]
So for the player $(i,\ell)$ in $\hat g$, \emph{every} alternative action strictly increases her payoff. In other words, $(i,\ell)$ is worst-responding at $\hat a$ in $\hat g$.
\end{proof}

We address a subtlety in the reduction. For any constants $m$ and $\hatm$, the function $f$ in \Cref{lem:gadget} can be implemented as a constant-size Boolean circuit. In our reduction, the number of actions $m=2$ and $\hatm$ of players in games $g$ and $\hat{g}$ are fixed, and for each possible $\hatm$ we can assume that the circuit gadget has constant size and is available as a precomputed circuit fragment of constant size. When our reduction constructs Boolean circuits for the payoff function $\hat u_{i\ell}$ of each player $(i,\ell)$ in $\hat{g}$, it thus takes constant computation time to insert each circuit gadget (a total number of $n$ times, one per group) into the overall circuit.

\begin{proof}[Proof of \cref{thm:potential-complexity}]
We first see that the reduction is polynomial time. The game $\hat{g}$ has $n q$ players and $\hatm$ actions. As $\hatm$ is a fixed constant and $q = 6\hatm 2^{\hatm}$ is a constant, the number of players of $\hat{g}$ is linear in the number $n$ of players in the original game $g$. The discussion above on the constant-size circuit gadget also makes it straightforward to see that the Boolean circuits for payoff functions $\hat u_{i\ell}$ of game $\hat{g}$ can be constructed in polynomial time. \Cref{lemma:g-hat-is-potential-game,lemma:game-admits-nwr-profile,lemma:pls-reduction-property} show that we have constructed a reduction. By \cref{fact:2-action-pot-games-hard}, the search problem for circuit games with $2$ actions per player is \PLShard, so it follows that the search problem of finding a no-worst-response action profile in a potential game with $\hatm$ actions is \PLShard.

Finally, since \Cref{lem:PLSmembership} established \PLS membership, the problem is \PLScomplete.
\end{proof}

\begin{remark}
\label{remark:identical-interest}
\Cref{thm:potential-complexity} in fact holds on the smaller class of \emph{identical-interest} games, those in which $u_i = u$ for every player $i$ and some common payoff function $u$. Every identical-interest game is a potential game, with $\phi = u$. The games of \Cref{fact:2-action-pot-games-hard} are already of this form, since the payoff function displayed there does not depend on $i$; the arguments establishing \PLS-membership in \Cref{lem:PLSmembership} apply to identical-interest games; and finally, \Cref{def:potential-reduction} preserves the property of identical interest, because if $u_i = u$ for every $i$ then $\hat u_{i\ell}(\hat a) = u(\rho(\hat a))$, which likewise does not depend on $(i,\ell)$.
\end{remark}

\subsection{The complexity of not-worst-responding in graphical games}\label{sec:graphical}

We now turn our attention to graphical games, which provide a complexity-theoretic wedge between pure Nash equilibrium and no-worst-response profiles. We consider a subclass of graphical games in which players have sufficiently many actions relative to the maximum vertex degree of the game's interaction graph. In this class, we show that pure Nash equilibrium existence is computationally intractable when each player has exactly $m \geq 28$ actions, whereas no-worst-response action profiles always exist.

\begin{definition}[Graphical games]
\label{def:graphical-game}
Let $g = (N, \{A_i\}_{i \in N}, \{u_i\}_{i \in N})$ be a game and let $G = (N,E)$ be an undirected graph on its players. For $i \in N$, write $N(i)$ for the set of neighbours of $i$ in $G$ and $N[i] \coloneqq N(i) \cup \{i\}$ for its closed neighbourhood, and write $\Delta(G) \coloneqq \max_{i \in N} |N(i)|$ for the maximum degree of $G$, abbreviated to $\Delta$ when the graph is clear from context. We call $G$ an \emph{interaction graph} for $g$ if each player's payoff depends only on the actions taken in her closed neighbourhood; that is, if for every $i \in N$ and all $a, b \in A$, if $a_j = b_j$ for every $j \in N[i]$ then $u_i(a) = u_i(b)$. A \emph{graphical game} is a pair $(g,G)$ in which $G$ is an interaction graph for $g$. For each integer $m \geq 2$, let $\Gamma_m$ denote the class of graphical games $(g,G)$ in which every player has exactly $m$ actions and
\[
    m \;\geq\; e\bigl(\Delta(G)^2 + 1\bigr) .
\]
\end{definition}

Throughout this section the input to our computational problems is always a pair $(g,G)$ in which $g$ is a circuit game, as defined in \Cref{sec:circuit-games}, together with the representation of an interaction graph $G$, which is supplied with the input.\footnote{Formally, write $\mathbb B(\Gamma_m)$ for the set of pairs
$(\langle g \rangle, \langle G \rangle)$ with $g \in \Games_m$ such that $G$ is an interaction graph for the game induced by $g$ and
$m \geq e(\Delta(G)^2+1)$. Then every computational problem considered in
\Cref{sec:graphical} has $\mathbb B(\Gamma_m)$ as its instance set, with
output sets $\{0,1\}$ (for decision problems) or $\mathbb B(\mathbf A_m)$ (for search problems).}

We first consider pure Nash equilibrium. Here, building on \citet{gottlob2003pure}, we see that the defining inequality for $\Gamma_m$ implies a sharp dichotomy at $m=28$.

\begin{proposition}
\label{prop:graphical-pne}
For every integer $m\geq 2$, the decision problem $m$-$\PNEexists$ of determining the existence of a pure Nash equilibrium in a graphical game in $\Gamma_m$ is in \Pclass if $m\leq 27$ and is \NPcomplete if $m\geq 28$.
\end{proposition}
For $m\geq 28$, it follows immediately that the corresponding search problem is \NPhard.

\begin{proof}
On the tractable side, fix $m \leq 27$ and a graphical game
$(g,G) \in \Gamma_m$. If $\Delta(G) \geq 3$, then membership in $\Gamma_m$
would require
\[
m \geq e(3^2+1) = 10e > 27,
\]
a contradiction. Hence $\Delta(G) \leq 2$ for every $(g,G) \in \Gamma_m$, and $G$ is a disjoint union of paths, cycles, and isolated vertices. Pure Nash equilibrium existence can be decided in polynomial time on each component by dynamic programming: proceeding along a path, it suffices to retain the actions of the two most recently visited players, since this determines whether the preceding player is best-responding; for a cycle, one may enumerate the actions of the first two players and apply the same procedure. The payoff comparisons required by the dynamic programme are computable in polynomial time from the payoff circuits. The procedures for the different components can be combined, yielding a polynomial-time algorithm that either constructs a pure Nash equilibrium or determines that none exists.

On the intractability side, fix $m \geq 28$. Membership in \NP is immediate. For hardness, we build on \citet{gottlob2003pure}, who show that deciding whether a graphical game admits a pure Nash equilibrium is \NPcomplete even when the interaction graph has maximum degree at most $3$ and every player has at most $3$ actions. These hard instances can be represented as circuit games in polynomial time.

Let $(g,G)$ be such a hard instance with $g=(N,\{A_i\}_{i\in N},\{u_i\}_{i\in N})$. For each player $i$ and action $a_i\in A_i$, choose positive integers $r_{i,a_i}$ such that
\[
    \sum_{a_i\in A_i}r_{i,a_i}=m.
\]
Such integers exist because $|A_i|\leq 3\leq m$. Construct a new graphical game $(g',G)$ on the player set $N$ by replacing each action $a_i\in A_i$ with $r_{i,a_i}$ clones
\[
    a_i^{(1)},\ldots,a_i^{(r_{i,a_i})}.
\]
Thus, every player in $g'$ has exactly $m$ actions, even though the players in the original game $g$ need not have had the same number of actions. Let $\pi_i$ map every clone of $a_i$ to $a_i$, and let $\pi$ apply these maps elementwise to action profiles. Define the payoff functions in $g'$ by
\[
    u_i'(a')=u_i(\pi(a')).
\]
Since $u_i'(a') = u_i(\pi(a'))$ and $\pi$ acts coordinatewise, $u_i'$ depends only on the actions in $N[i]$, so $G$ remains an interaction graph for $g'$, and $\Delta(G) \leq 3$ is unchanged. Moreover 
\[
m \geq 28 \geq e(3^2+1),
\]
and hence $(g',G) \in \Gamma_m$.

It remains to show that pure Nash equilibrium existence is preserved. Suppose first that $a$ is a pure Nash equilibrium of $g$. Choosing an arbitrary clone of each action in $a$ yields an action profile $a'$ of $g'$ with $\pi(a')=a$. Any unilateral deviation in $g'$ projects either to the same action in $g$ and leaves the deviating player's payoff unchanged, or to a unilateral deviation in $g$, which cannot be profitable. Thus, $a'$ is a pure Nash equilibrium of $g'$.

Conversely, let $a'$ be a pure Nash equilibrium of $g'$ and set $a=\pi(a')$. If some player had a profitable deviation from $a_i$ to an action $b_i$ in $g$, then deviating in $g'$ to any clone of $b_i$ would give the same strict payoff improvement. This would contradict the assumption that $a'$ is a pure Nash equilibrium. Hence, $a$ is a pure Nash equilibrium of $g$.

We have therefore constructed, in polynomial time, a graphical game $(g',G)\in\Gamma_m$ that admits a pure Nash equilibrium if and only if $g$ does. This establishes \NP-hardness for every fixed $m\geq 28$ and completes the proof.
\end{proof}

Next, we show that the decision problem $m$-$\NWRexists$ is trivial for any graphical game in $\Gamma_m$ for every $m \geq 2$, thus providing the desired wedge in complexity.

\begin{theorem}
\label{thm:graphical-nwr}
For every integer $m\geq 2$, every graphical game in $\Gamma_m$ admits a no-worst-response profile.
\end{theorem}

\begin{proof}
Fix an integer $m\geq 2$ and a graphical game $(g,G)\in\Gamma_m$. Choose an action profile uniformly at random, say $a\in A$. For each player $i\in N$, define the bad event $B_i$ that player $i$ plays a worst response in $a$.

Fix a player $i$. Given fixed actions $a_{-i}$ for all other players, there can be at most one strict worst action for player $i$: if two distinct actions are tied for minimum payoff, neither is a strict worst response. Hence,
\[
    \Pr[B_i]\leq \frac{1}{m}.
\]

The event $B_i$ depends only on the actions in the closed neighbourhood $N[i]:=N(i)\cup\{i\}$. Thus, $B_i$ is independent of all $B_j$ such that $N[i]\cap N[j]=\emptyset$, that is, of all players $j$ at graph distance at least $3$ from $i$. In a graph of maximum degree $\Delta$, the number of players at distance exactly $1$ from $i$ is at most $\Delta$, and the number of players at distance exactly $2$ from $i$ is at most $\Delta(\Delta-1)$. Hence, the number of players distinct from $i$ at distance at most $2$ from $i$ is at most
\[
    \Delta+\Delta(\Delta-1)=\Delta^2.
\]
Therefore, each event $B_i$ depends on at most $D\leq\Delta^2$ other events. By the symmetric Lov\'{a}sz Local Lemma \citep[see, e.g.,][]{alon-spencer}, we know that with positive probability no bad event occurs at $a$, as long as
\[
    e\cdot\frac{1}{m}\cdot(D+1)\leq 1.
\]
This holds as $m\geq e(\Delta^2+1)$ by assumption. Hence, there exists an action profile with no worst responses.
\end{proof}

\Cref{thm:graphical-nwr} provides a sharp contrast with \cref{prop:graphical-pne}. For every $m\geq 28$, no-worst-response action profiles are guaranteed to exist in graphical games in $\Gamma_m$, whereas deciding whether a pure Nash equilibrium exists is \NPcomplete. We also note that if the graphical structure of the game is known and worst responses can be identified in polynomial time, then the proof of \cref{thm:graphical-nwr} yields a randomised polynomial-time algorithm for finding a no-worst-response profile via the algorithmic Lovász Local Lemma of \citet{moser-tardos}.

\section{Non-universal guarantees}
\label{section:non-universal-guarantees}

Thus far we considered minimal rationality guarantees for \emph{all players simultaneously} and established that finding no-worst-response action profiles is broadly as ``hard'' as finding pure Nash equilibria. We next turn to the question of whether action profiles where only some fraction ($<1$) of players are guaranteed to play ``reasonably good'' actions are computationally tractable.

An action profile $a \in A$ is an \emph{$\albe$-action profile} if a fraction at least $\lceil \alpha n \rceil/n$ of players play a top-$\lceil \beta m\rceil/m$ action in response to the other players' actions; that is, the fraction of actions $a_i'$ for which $u_i(a_i',a_{-i}) > u_i(a_i,a_{-i})$ is less than $\lceil \beta m\rceil/m$.
For \emph{non-universal guarantees} ($\NUG$), we define the decision problem $m$-$\NUGexists$ and the counting problem $m$-$\NUGcount$ of deciding the existence of, and counting, $\albe$-action profiles, in analogy to \cref{def:decide,def:count}; the search problem $m$-$\NUGfind$, our main object of study, is defined below (\cref{def:NUGfind}).

We identify a threshold, where on one side, the problem becomes tractable and, on the other side, the problem remains intractable. 

\begin{definition}[Search problem: $m$-$\NUGfind$]
Fix an integer $m \geq 2$. An instance of the search problem $m$-$\NUGfind$ consists of a triple $(\langle g \rangle, \langle \alpha \rangle, \langle \beta \rangle)$, where $g \in \Games_m$ is a circuit game with $n$ players, each of whom has exactly $m$ actions, and where $\alpha = k/n$ and $\beta = \ell/m$ for some $k \in \{1, \ldots, n\}$ and $\ell \in \{1, \ldots, m\}$. Writing $\In_m$ for the set of all such instances, the problem of finding an $\albe$-action profile is the relation
\[
    m\text{-}\NUGfind \subseteq \In_m \times \mathbb B(\mathbf A_m),
\]
where $\big((g,\alpha,\beta), a\big) \in m\text{-}\NUGfind$ if and only if $a$ is an $\albe$-action profile of $g$. An algorithm solves $m$-$\NUGfind$ if, on input $(g, \alpha, \beta) \in \In_m$, it outputs such an action profile or returns ``infeasible'' if none exists.
\label{def:NUGfind}
\end{definition}

\begin{theorem}
\label{thm:threshold}
For every integer $m \geq 2$, the search problem $m$-$\NUGfind$, restricted to instances $(g, \alpha, \beta)$ with $\alpha \leq \beta$, admits a randomised polynomial-time algorithm that succeeds with probability at least $1/2$, and, restricted to instances with $\alpha > \beta$, is \textnormal{\NPhard}.
\end{theorem}

By \cref{thm:threshold}, on the one hand, finding an action profile where a fraction at least $\alpha$ of players best-respond is tractable if $\alpha\leq 1/m$; on the other hand, finding an action profile where a fraction $\alpha$ of players do not worst-respond is tractable if $\alpha\leq (m-1)/m$ (see \cref{fig:separation}). This shows that our negative results regarding action profiles where no player worst-responds are in some sense knife-edge. If the number of actions per player is large, an action profile where almost nobody worst-responds is tractable. By contrast, an action profile where a constant fraction of players best-respond is intractable.

\begin{figure}[H]
\centering
\begin{tikzpicture}
\begin{groupplot}[
    group style={
        group size=2 by 1,
        horizontal sep=2.5cm,
    },
    width=7cm,
    height=6cm,
    xmin=2, xmax=10,
    ymin=0, ymax=1,
    xtick={2,3,4,5,6,7,8,9,10},
    ytick={0,.5,1},
    axis lines=left,
    axis line style={-},
    xlabel={Number of actions per player},
]
\nextgroupplot[
    ylabel={},
    title={Fraction of players best-responding},
]
\addplot [
    domain=2:10,
    samples=200,
    fill=gray!25,
    draw=none
] {1/x} \closedcycle;
\addplot [
    domain=2:10,
    samples=200,
    thick
] {1/x};
\node at (axis cs:5.7,0.06) {\small randomised poly-time search};
\node at (axis cs:6,0.5) {\small \NPhard};
\nextgroupplot[
ylabel={},
    title={Fraction of players not worst-responding},
    yticklabels={0,0.5,1},   
    axis y line=left,
    axis line style={-},
      ylabel style={rotate=0, anchor=center}
]
\addplot [
    domain=2:10,
    samples=200,
    fill=gray!25,
    draw=none
] {(x-1)/x} \closedcycle;
\addplot [
    domain=2:10,
    samples=200,
    thick
] {(x-1)/x};
\node at (axis cs:6,0.5) {\small randomised poly-time search};
\node at (axis cs:4,0.9) {\small \NPhard};
\end{groupplot}
\end{tikzpicture}
\caption{Tractability frontier for two top-$\beta$ guarantees. The left panel considers best responses, corresponding to $\beta=1/m$; the right panel considers no-worst-response profiles, corresponding to $\beta=(m-1)/m$. In the shaded region, the corresponding search problem admits a randomised polynomial-time algorithm that returns a valid action profile with probability at least $1/2$; above the curve, it is \NPhard.}
\label{fig:separation}
\end{figure}

The proof of the first part of \cref{thm:threshold} follows from \cref{prop:random-guarantee} below. The second part of \cref{thm:threshold} follows from \cref{prop:finer-hardness} that shows that the corresponding decision problem is \NPcomplete and hence the search problem is \NPhard.

The intuition for the randomised-algorithm part of \cref{thm:threshold} is the following. Consider the boundary where $\alpha=\beta$. If each player selects an action at random, with probability $\alpha$ they will play one of their approximately top $\alpha \cdot m$ actions. Then, by linearity of expectation, when a random action profile is drawn, the expected number of players who are using one of their approximately top $\alpha \cdot m$ actions is $\alpha \cdot n$. We formalise this intuition to any fixed number $m$ of actions per player and for $\alpha, \beta$ that do not need to be fixed a priori.

\begin{algorithm}
    \begin{algorithmic}[1]
    \Require An instance $(g, \alpha, \beta) \in \In_m$ of $m$-$\NUGfind$ with $\alpha \leq \beta$, where $g = (N, \{A_i\}_{i \in N}, \{u_i\}_{i \in N})$ and $n \coloneqq |N|$.
    \Ensure An action profile $a \in A$ that, with probability at least $1/2$,
    is an $\albe$-action profile of $g$.
    \State Let $t = 1$ and $a$ be an arbitrary action profile.
        \While {$t \leq n+1$}
            \State Let $c = 0$.
            \State Let $a$ be an action profile drawn uniformly at random from $A$.
            \For{each player $i \in N$}
                \If{$a_i$ is a top $\beta$ action in response to $a_{-i}$}
                \State Let $c = c+1$
                \EndIf
            \EndFor
            \If{$c \geq \alpha n$}
            \State \Return action profile $a$.
            \EndIf
            \State Let $t = t+1$.
        \EndWhile
        \State \Return action profile $a$ from the loop's last iteration.
    \end{algorithmic}
    \caption{A false-biased Monte Carlo algorithm.}
    \label{alg:monte-carlo}
\end{algorithm}

\begin{proposition}
\label{prop:random-guarantee}
Fix an integer $m \geq 2$. For every instance $(g, \alpha, \beta) \in \In_m$ of $m$-$\NUGfind$ with $\alpha \leq \beta$, \cref{alg:monte-carlo} returns an $\albe$-action profile of $g$ with probability at least $1/2$.
\end{proposition}

\begin{proof}[Proof of \Cref{prop:random-guarantee}]
For any action profile $a \in A$, let $\omega(a)$ denote the number of players whose action at $a$ is a top $\beta$ action. Since $\alpha = k/n$ for some $k \in \{1, \ldots, n\}$, the quantity $\alpha n$ is an integer, and $a$ is an $\albe$-action profile if and only if $\omega(a) \geq \alpha n$. Since \cref{alg:monte-carlo} samples $n+1$ action profiles independently, the probability that it returns a profile satisfying $\omega(a) \geq \alpha n$ is
\[
    1 - \mathbb{P}[\, \omega(a) < \alpha n \,]^{n+1},
\]
where $\mathbb{P}[\, \omega(a) < \alpha n \,]$ denotes the probability that $a \sim \mathrm{Unif}[A]$ satisfies $\omega(a) < \alpha n$. But
\[
\mathbb{P}[\, \omega(a) < \alpha n \,]
= \mathbb{P}[\, \omega(a) \leq \alpha n - 1 \,]
= \mathbb{P}[\, n - \omega(a) \geq n + 1 - \alpha n \,]
\leq \frac{n - \mathbb{E}[\omega(a)]}{n + 1 - \alpha n},
\]
where the last step is an application of Markov's inequality.

Now, as $\beta = \ell/m$ for some $\ell \in \{1, \ldots, m\}$, at least $\ell$ of each player's $m$ actions are top $\beta$ actions, so a uniformly random action of player $i$ is a top $\beta$ action with probability at least $\beta$. By linearity of expectation, $\mathbb{E}[\omega(a)] \geq \beta n$. Together with $\alpha \leq \beta$, this yields
\[
    \frac{n - \mathbb{E}[\omega(a)]}{n + 1 - \alpha n}
    \leq \frac{n - \beta n}{n + 1 - \beta n}
    \leq \frac{n}{n+1} .
\]
It follows that the probability that \cref{alg:monte-carlo} returns an $\albe$-action profile is at least
\[
    1 - \left( \frac{n}{n+1} \right)^{n+1} \geq 1 - \frac{1}{e} > \frac{1}{2}
\]
for any $n \geq 1$, which completes the proof.\footnote{Line 4 of \cref{alg:monte-carlo} specifies a uniform draw from $A$, which for $m$ not a power of two differs from a uniform draw of $ns$ bits. A uniform action is obtained by rejection sampling: draw $s$ bits and redraw if the value is at least $m$. Since $s = \lceil\log_2 m\rceil$ gives $2^s < 2m$, each attempt is accepted with probability $m/2^s > 1/2$. We truncate at $T = \lceil \log_2(16\,n(n+1))\rceil$ attempts per action, returning $\bot$ if any draw is rejected $T$ times in succession; a run makes $n(n+1)$ action draws, so by a union bound this occurs with probability at most $1/16$. Coupling the truncated procedure with the idealised one on the same random bits, the two agree unless a truncation occurs, so the success probability falls short of the bound established in \Cref{prop:random-guarantee} by at most $1/16$. Since that bound is $1 - 1/e > 0.63$, the implemented algorithm still succeeds with probability at least $1/2$, as claimed.}
\end{proof}

We next turn to the hardness result in \cref{thm:threshold}.

\begin{proposition}
\label{prop:finer-hardness}
For every integer $m \geq 2$, the decision problem $m$-$\NUGexists$, restricted
to instances $(g, \alpha, \beta)$ with $\alpha > \beta$, is \NPcomplete.
\end{proposition}
To motivate the proof, consider, for any $n = m \geq 2$, a game with $n$ players in which the players' actions are $\{0, \ldots, m-1\}$, and suppose that their utilities are given by $u_i(a) = \left(i + \sum_{j \in N} a_j \right) \mod{m}$. At every action profile of the game, for each $q \in \{1,\ldots,m\}$, exactly $q$ players are using one of their $q$ best actions. It turns out that this property is preserved with the addition of the global term of \eqref{eq:payoffs} in \cref{def:circuitSAT-reduction}. 

For the proof of \cref{prop:finer-hardness} and the discussion in the following, write $\beta=\ell/m$ and let
\[
    r=\max\{t\in\{0,\ldots,m-1\}: \alpha>t/m\}.
\]
Since $\beta<\alpha$ and $\beta=\ell/m$, we have $\ell\le r$, and hence every top-$\beta$ action is among the player's top $r$ actions.

Membership in \NP is immediate, as we can guess an action profile and check whether at least a fraction $\alpha$ of players are playing a top-$\beta$ action in polynomial time. For \NP-hardness, we argue that \cref{def:circuitSAT-reduction} is a polynomial-time reduction from $\circuitSATexists$ to $m$-$\NUGexists$. The main challenge to proving that \cref{def:circuitSAT-reduction} is a polynomial-time reduction to $m$-$\NUGexists$ is to provide a more granular version of \cref{obs:circuitSAT-reduction}.

\begin{observation}
\label{obs:circuitSAT-reduction-granular}
If $f(a) = 1$, then every player is playing their best action. If $f(a) = 0$, then at most a fraction $r/m$ of players are playing one of their top $r$ actions, and therefore at most a fraction $r/m$ of players are playing a top-$\beta$ action.
\end{observation}
\begin{proof}[Proof of \Cref{obs:circuitSAT-reduction-granular}]
If $f(a) = 1$, then \cref{obs:circuitSAT-reduction} tells us that every player is playing their best action. So in particular, all $nm$ players are using one of their $r$ best actions.

So suppose from now on that $f(a) = 0$. If we can show that at most $nr$ players are using one of their top $r$ actions, then, since $\ell \le r$, at most $nr$ players are using a top-$\beta$ action. Thus the fraction of players using a top-$\beta$ action is at most $r/m<\alpha$. Hence the only action profiles $a$ in which at least an $\alpha$ fraction of players play a top-$\beta$ action are those for which $f(a)=1$.

Now suppose no player can unilaterally deviate to force $f(a'_{i\ell}, a_{-i\ell}) = 1$. Thus, by construction of the local term in \eqref{eq:payoffs}, each player's mapping from actions to payoffs is a permutation on $\{0, \ldots, m-1 \}$. Secondly, within any group of players $(i, 1), \ldots, (i, m)$, for each possible payoff $p \in \{0, \ldots, m-1\}$ there exists a player $(i,\ell)$ who achieves payoff $p$. It follows that exactly $r$ players in each group play one of their $r$ best actions. As there are $n$ groups, $nr$ players play one of their $r$ best actions in total.

\newcommand{\hati}{j}
\newcommand{\hatell}{\ell'}

Finally, suppose there exists a player $(\hati, \hatell)$ who can deviate to force $f(a'_{\hati\hatell}, a_{-\hati\hatell}) = 1$. We notice that there is always at most one group $\hati$ which contains players who can achieve this. Also, by the same argument as above, exactly $r$ players in each of the other groups $i \neq \hati$ play one of their $r$ best actions. So we focus on counting the number of players in the group $(\hati,1), \ldots, (\hati,  m)$ (which we call group $\hati$) who are playing one of their $r$ best actions. If $m = 2$ and if both players $(\hati,1),(\hati,2)$ may deviate to force $f(a'_{\hati\hatell}, a_{-\hati\hatell}) = 1$, then they are both worst responding and we are done. If instead $m \geq 3$, there is necessarily only one such player $(\hati,\hatell)$. We thus focus on the case when there exists a unique such player $(\hati,\hatell)$, and distinguish between two subcases. Recall that the local term of player $(\hati,\hatell)$'s payoff is
\[
r_{\ell'} =
\left(\hatell+\sum_{k=1}^m a_{jk}\right) \bmod m,
\]
and the value of this term can range from $0$ to $m-1$ (cf.~\cref{eq:payoffs} in \cref{def:circuitSAT-reduction}).

\textbf{Case 1:} Suppose $r_{\ell'} < m-r$. Then exactly $r$ of the remaining players in group $\hati$ are playing one of their $r$ best actions, by the same argument as above. Player $(\hati, \hatell)$ can improve over her current action $a_{\hati\hatell}$ by playing the action that flips $f(a)$ to $1$, or playing one of the $r$ actions for which her local term is greater than or equal to $m - r$. The flipping action and the $r$ actions may intersect but, regardless, the action $a_{\hati\hatell}$ is not among her $r$ best actions. It follows that exactly $r$ players in group $\hati$ play one of their $r$ best actions in action profile $a$.

\textbf{Case 2:} Suppose $r_{\ell'} \geq m-r$. By the same argument as above, there exist exactly $r-1$ other players in group $\hati$ who are playing one of their $r$ best actions. So, together with $(\hati, \hatell)$, at most $r$ players are playing one of their $r$ best actions in profile $a$.
\end{proof}

With this observation, the proof of \cref{prop:finer-hardness} is now straightforward.
\begin{proof}[Proof of \cref{prop:finer-hardness}]
Let $C$ be a Boolean circuit with $n$ inputs. If $C$ is satisfiable, then there exists an action profile for the game with $f(a) = 1$, and by \cref{obs:circuitSAT-reduction-granular} all $nm$ players play their best response, hence a top-$\beta$ action. If $C$ is unsatisfiable, then $f(a) = 0$ for all action profiles, and so \cref{obs:circuitSAT-reduction-granular} tells us that at most an $r/m < \alpha$ fraction of players play a top-$\beta$ action. It follows that $C$ is a yes-instance of $\circuitSATexists$ if and only if $(g,\alpha,\beta)$ is a yes-instance of $m$-$\NUGexists$.
\end{proof}

We can also extend the hardness result of \cref{prop:finer-hardness} to the counting version of the problem. 

\begin{corollary}\label{cor:last}
For every integer $m \geq 2$, the counting problem $m$-$\NUGcount$, restricted to instances $(g, \alpha, \beta)$ with $\alpha > \beta$, is \textnormal{\Pcounthard}.
\end{corollary}
\begin{proof}[Proof of \Cref{cor:last}]
We use the same definition of $r$ as in the proof of \cref{prop:finer-hardness}. By \cref{obs:circuitSAT-reduction-granular}, an action profile in the reduced game $g$ has at least an $\alpha$ fraction of the $nm$ players playing a top-$\beta$ action if and only if $f(a)=1$. It is straightforward that there is a one-to-one correspondence between satisfying truth assignments of $C$ and action profiles $a$ with $f(a)=1$. Hence, $m$-$\NUGcount$ is \Pcounthard.
\end{proof}

\section{Discussion}

We have shown that even the weakest meaningful individual rationality requirement---that no player plays a worst response---is computationally intractable if required for all players in games with fixed numbers of actions. In general games, the decision, search, and counting problems for no-worst-response action profiles are computationally equivalent to their pure Nash equilibrium counterparts (\NPcomplete, \NPhard, and \Pcountcomplete, respectively). This hardness persists in potential games, where finding a no-worst-response action profile is \PLScomplete despite the simple structure of potential games and the guaranteed existence of pure Nash equilibria and thus of no-worst-response action profiles. 

Our results suggest that, while not playing one's worst action may be straightforward in a single agent setting, requiring such behaviour for all players in a game turns out to be prohibitively difficult. This carries implications for what outcomes one can expect to arise in games. Solution concepts in games are usually founded on the premise that they would be played by rational players, or that an analyst recommends an action profile to which the players (after consideration) adhere, or as an implicit or explicit agreement found by players after communication, or as arising from a decentralised learning procedure. Our results imply that, regardless of which of these foundations one considers, either some player or the analyst has to overcome a computationally intractable problem if the solution concept implements only no-worst-response action profiles. This computational intractability violates a basic plausibility requirement of decision-making \citep{camara}. 

From this perspective, our results are sobering. Two qualifications are in order, however, and each points to avenues for future research.

Firstly, our hardness results concern the regime in which the number of actions per player is bounded while the number of players grows. Indeed, when each player has only two actions, no-worst-response action profiles and pure Nash equilibria coincide, and our results show that the two concepts remain broadly computationally equivalent for any fixed number of actions. 
In the complementary regime, in which the number of players is held fixed and action sets grow, \citet[Proposition~3]{pradelski2024satisficing} show that whenever players have more actions than there are players ($m > n$), every game admits a no-worst-response action profile. 
This may lead one to conjecture, broadly, that no-worst-response action profiles are computationally tractable in this regime, which merits further inquiry.

Secondly, it is important to note that we considered worst-case complexity; that is, the maximal complexity of a given problem across all instances in a given game class. A natural route forward is to consider less penalising notions. One such route is to consider average-case complexity, i.e., the expected computational complexity of a problem when instances are drawn from some distribution. Relatedly, the convergence of simple learning dynamics may, on average, be much quicker for no-worst-response action profiles than for pure Nash equilibrium.

\singlespacing{\small
\bibliographystyle{plainnat}
\bibliography{references}
}

\newpage
\onehalfspacing
\appendix

\section{Primer on computational complexity}
\label{sec:appendix}

\paragraph{Boolean circuits.} 
A \emph{Boolean circuit} $C$ represents a multivariate function that takes $n$ binary inputs, $(x_1, \ldots, x_n)$, with $1$ representing \emph{true} and $0$ representing \emph{false}. The circuit has one or more outputs, which are computed through a sequence of three types of elementary logical gates: AND ($\land$), OR ($\lor$), and NOT $(\lnot)$. The circuit forms a directed acyclic graph: the $n$ input bits feed into gates, whose outputs feed into subsequent gates, until the final gates produce the circuit's outputs. The fan-in of each gate is at most two. See \Cref{fig:circuit} for an example with one output gate. The \emph{size} $|C|$ of the circuit $C$ is the number of its gates. We write $\tau \in \{0,1\}^n$ for a truth assignment to the inputs $x_1, \ldots, x_n$, and $C(\tau)$ for the output of circuit $C$ given input $\tau$. When $C$ is a circuit with one output, we say that a truth assignment $\tau$ \emph{satisfies} $C$ if $C(\tau) = 1$.
 
\begin{figure}[h]
\centering
\begin{tikzpicture}[
    xscale=2.5,
    yscale=1.5,
    gate/.style={circle, draw, thick, minimum size=0.8cm, inner sep=3pt, outer sep=3pt},
    input/.style={coordinate}
]

\node[input] (x1) at (0, 2) {};
\node[input] (x2) at (0, 1) {};
\node[input] (x3) at (0, 0) {};

\node[left=0.1cm of x1] {$x_1$};
\node[left=0.1cm of x2] {$x_2$};
\node[left=0.1cm of x3] {$x_3$};

\node[gate] (g1) at (1, 2) {$\land$};
\node[gate] (g2) at (1, 1) {$\lor$};
\node[gate] (g3) at (1, 0) {$\lnot$};

\node[gate] (g4) at (2, 1.5) {$\lnot$};
\node[gate] (g5) at (2, 0.5) {$\land$};

\node[gate] (g6) at (3, 1) {$\lor$};

\node[] (g7) at (4, 1) {$y$};

\path[->] (x1) edge (g1);
\path[->] (x2) edge (g1);
\path[->] (x2) edge (g2);
\path[->] (x3) edge (g2);
\path[->] (x3) edge (g3);

\path[->] (g1) edge (g4);
\path[->] (g2) edge (g5);
\path[->] (g3) edge (g5);

\path[->] (g4) edge (g6);
\path[->] (g5) edge (g6);

\path[->] (g6) edge (g7);

\path ([shift={(-9pt,-9pt)}]current bounding box.south west)
      ([shift={( 9pt, 9pt)}]current bounding box.north east);
\end{tikzpicture}
\caption{A Boolean circuit of size $6$ with three inputs $(x_1, x_2, x_3)$ and one output. The gates are $\land$ (AND), $\lor$ (OR), and $\lnot$ (NOT). Truth assignment $\tau = (1,1,1)$ evaluates to $0$ (false), while assignment $\tau' = (1,1,0)$ evaluates to $1$ (true).}
\label{fig:circuit}
\end{figure}

\paragraph{Computational problems.}
Algorithms take binary strings as inputs and return binary strings as outputs, and we define our computational problems in those terms. Let $\{0,1\}^*$ denote the set of all binary strings. A \emph{computational problem} is a relation
\[
    f \subseteq \In \times \Out ,
\]
where $\In \subseteq \{0,1\}^*$ is a set of \emph{instances} and $\Out \subseteq \{0,1\}^*$ a set of \emph{candidate outputs}. An \emph{algorithm} is a procedure that, given an instance $x \in \In$, performs a finite sequence of elementary steps and halts, returning either a string $y \in \Out$ or a distinguished symbol $\bot$, read as the report ``no such output exists'' or ``infeasible''. The algorithm \emph{solves} $f$ if, for every $x \in \In$, it returns some $y$ with $(x,y) \in f$ whenever such a $y$ exists, and returns $\bot$ whenever none does. The \emph{size} of an instance $x$ is its length $|x|$, and the \emph{running time} of an algorithm on $x$ is the number of elementary steps it performs; we say that $f$ is `solvable' or `computable' in polynomial time if some algorithm solves it in a number of steps bounded by a polynomial in $|x|$. Nothing in this article turns on the choice of a formal model of computation (e.g. Turing machines), since the standard models agree on which problems are solvable in polynomial time.

\paragraph{Reductions.}
Reducing a problem $f$ to a problem $f'$ means exhibiting a way of turning any instance of $f$ into an instance of $f'$, and turning any solution of the latter back into a solution of the former, at polynomial cost. If such a reduction exists, then $f$ is no harder than $f'$: an efficient algorithm for $f'$ yields one for $f$.
 
\begin{definition}[Polynomial-time reductions]
\label{def:reduction}
Let $f \subseteq \In \times \Out$ and $f' \subseteq \In' \times \Out'$ be
computational problems. A \emph{polynomial-time reduction} from $f$ to $f'$
is a pair $(\push,\pull)$ of polynomial-time computable functions
\[
    \push : \In \to \In' ,
    \qquad
    \pull : \In \times \Out' \to \Out ,
\]
such that for every instance $x \in \In$: (i) for every $y' \in \Out'$, if $(\push(x),y') \in f'$ then $\bigl(x,\pull(x,y')\bigr) \in f$, and (ii) if $(x,y) \in f$ for some $y \in \Out$, then $(\push(x),y') \in f'$ for some $y' \in \Out'$. We then say that $f$ \emph{reduces} to $f'$. Figure \ref{fig:flow-diagram} gives a sketch.
\end{definition}

Condition (i) says that solutions transfer back, and condition (ii) that
solvability transfers forward (so an algorithm for $f'$ that correctly
reports ``infeasible'' can be used to answer ``infeasible'' for $f$).

\begin{figure}[h]
    \centering
    \begin{tikzpicture}[>=stealth]
        \node (A) at (0,2.4)   {$\In$};  
        \node (X) at (5.4,2.4) {$\Out$};
        \node (B) at (0,0)     {$\In'$}; 
        \node (Y) at (5.4,0)   {$\Out'$};

        \draw[->] (A) -- node[above] {$f$}     (X);
        \draw[->,thick,blue] (A) -- node[left]  {$\push$} (B);
        \draw[->] (B) -- node[below] {$f'$}    (Y);

        \coordinate (M) at (5.4,0.8);
        \draw[thick,red] (A.south east) to[out=-15, in=270, in looseness=0.6] (M);
        \draw[->,,thick,red] (Y) -- node[right] {$\pull$} (X);

    \end{tikzpicture}
    \caption{A reduction from $f$ to $f'$. The map $\push$ sends an instance of $f$ to an instance of $f'$; the map $\pull$ takes an instance of $f$ together with a solution of the transformed instance, and returns a solution of the original, which is why two lines merge into the arrow carrying its label. Two caveats about the drawing. The arrows labelled $f$ and $f'$ are drawn as though they were functions, which is for legibility only: $f$ and $f'$ are relations, so an instance may have many solutions or none. And the figure depicts condition~(i) of \Cref{def:reduction}; condition~(ii), which requires $\push(x)$ to have a solution whenever $x$ does, has no counterpart in the picture.}
    \label{fig:flow-diagram}
\end{figure}

For most of our results, we establish hardness by reducing a known hard problem to the corresponding problem for $\NWR$. In the case of decision problems, $\Out = \Out' = \{0,1\}$ and each instance has exactly one solution, so (ii) is vacuous. In turn, taking $\pull(x,y') = y'$, (i) says that $\push$ maps yes-instances to yes-instances and no-instances to no-instances, recovering the usual many-one (Karp) reduction. In the case of search problems, $\pull$ converts a solution of the constructed instance into a solution of the original one. In the case of counting problems, we need slightly more. A reduction between search problems is \emph{parsimonious} if, in addition, $\pull(x,\cdot)$ is a bijection from the set of solutions of $\push(x)$ onto the set of solutions of $x$. A parsimonious reduction preserves the number of solutions and therefore induces a reduction between the corresponding counting problems.

Two observations let us argue about reductions at the level of games and action profiles, rather than at the level of their encodings. First, each $\push$ we construct is described by exhibiting a game, and the encoding of that game is then written down directly; it is therefore automatic both that $\push$ lands in the target instance set and that it is computable in polynomial time, provided the circuits exhibited are of polynomial size, which is what we verify in each case. Second, because the binary encodings $\langle\cdot\rangle$ of \Cref{sec:compproblems} are injective, a bijection between two sets of solutions induces a bijection between the corresponding sets of encodings; so to establish that a reduction is parsimonious it suffices to exhibit the bijection between solutions. Accordingly, from \Cref{sec:NWRcomplexity} onwards we have suppressed the encoding brackets entirely and have spoken of games and profiles themselves, reverting to strings only where the length of an encoding is at issue.

\paragraph{Complexity classes.} 
We here describe the complexity classes used in our article.

 \Pclass is the class of decision problems solvable in polynomial time, and is the standard formalisation of computational tractability. \NP is the class of decision problems whose yes-instances admit short and quickly checkable evidence. So solutions may be hard to compute but are easily verifiable. Concretely, a decision problem lies in \NP if there is a polynomial-time algorithm that takes an instance $x$ together with a second string $c$ of length polynomial in $|x|$, and accepts or rejects the pair, in such a way that some $c$ is accepted alongside $x$ exactly when $x$ is a yes-instance. Such a $c$ is called a \emph{certificate} for $x$. A problem is \NPhard if every problem in \NP reduces to it, and \NPcomplete if it is in addition a member of \NP. So to show that a problem $f'$ is \NPcomplete, it suffices to show that the problem is in \NP and that there is a polynomial-time reduction from a known \NPcomplete problem $f$ to $f'$. For instance, the Boolean satisfiability problem \cite{cook1971complexity} (hence, $\circuitSATexists$) as well as the decision version of the Travelling Salesman problem \cite{karp1972reducibility} are canonical \NPcomplete problems.

\Pcount is the counting analogue of \NP: rather than asking whether a
solution exists, one asks how many there are. A counting problem is \Pcounthard if every problem in \Pcount reduces to it by a polynomial-time parsimonious (count-preserving) reduction, and \Pcountcomplete if it is in addition a member of \Pcount. A canonical
\Pcountcomplete problem is $\circuitSATcount$, which asks for the number of
satisfying assignments of a Boolean circuit \citep{valiant1979enumeration}.

\PLS (``polynomial local search'') is the class of problems of the
following shape \citep{johnson1988easy}: an instance of a problem in \PLS  consists of a finite set of feasible solutions, each with a polynomial-length encoding, an integer-valued objective function whose values have polynomial-length representations, and a neighbourhood relation. There must exist a polynomial-time algorithm that
\begin{enumerate}[leftmargin=*, labelindent=0pt, label=(\roman*)]
    \item produces an initial feasible solution;
    \item evaluates the objective at a given feasible solution; and
    \item given a feasible solution, either returns a neighbouring solution
    with strictly greater objective value, or certifies that no such
    neighbour exists.
\end{enumerate}
A solution admitting such a certificate is a \emph{local optimum}, and the task is to find one. Because the feasible set is finite, repeatedly taking improving steps eventually reaches a local optimum, so a solution always exists (but the number of steps required may nonetheless be exponential).

A problem is \PLShard if every problem in \PLS reduces to it, and \PLScomplete if it is in addition a member of \PLS \citep{johnson1988easy}. Finding a pure Nash equilibrium in a succinctly represented potential game is a canonical \PLScomplete problem \citep{FabrikantPT04} (a pure Nash equilibrium is exactly a local maximum of the potential function and best-response dynamics is exactly local search). \PLS-hardness rather than \NP-hardness is the appropriate complexity class for settings where existence is guaranteed (as it is for Nash equilibria in potential games): the solution is hard to find precisely because searching for it, and not deciding whether it exists, is the difficulty.

\RPclass is the class of decision problems admitting a polynomial-time randomised algorithm with one-sided error: it rejects every no-instance and accepts every yes-instance with probability at least $1/2$ \citep[see, e.g.,][]{Mitzenmacher_Upfal_2005}. The analogue for a search problem, which is what we use and which we also call \RPclass, is the class of search problems admitting an algorithm that runs in polynomial time and, on every instance admitting a solution, outputs a valid solution with probability at least $1/2$. Our positive results (\Cref{thm:graphical-nwr,thm:threshold}) are about problems in this class.

\end{document}